\documentclass[a4paper,cleveref,USenglish]{lipics-v2021}

\usepackage{auxlib}

\usepackage[utf8]{inputenc}
\usepackage{quiver}%for the diagram ordering the tuples
\usepackage{xargs}%more than one optional parameter in a new commands, for todonotes
\usepackage[colorinlistoftodos,prependcaption,textsize=small]{todonotes}

\usepackage{xspace}
\usepackage{listings}
\usepackage{framed}
\usepackage{float}
\usepackage{tikz}
\usetikzlibrary{patterns}
\MakeOuterQuote{"}
\makeatletter
\def\HyPsd@CatcodeWarning#1{}
\makeatother

\newcommand*{\BigO}[1]{\ensuremath{\mathchoice{\LDAUOmicron{#1}}{\smash{\LDAUOmicron{#1}}}{\LDAUOmicron{#1}}{\LDAUOmicron{#1}}}}

\def\Vardef#1{%
	\expandafter\newcommand\csname #1\endcsname[1]{%
		\def\first{##1}%
		\def\second{*}%
		\def\third{}%
		\ensuremath{\mathsf{\MakeLowercase #1}\ifx\first\second\else\ifx\first\third\else[{##1}]\fi\fi}%
	}%
}
\def\Vardefxx#1#2{%
	\expandafter\newcommand\csname #1\endcsname[1]{%
		\def\first{##1}%
		\def\second{*}%
		\def\third{}%
		\ensuremath{\mathsf{#2}\ifx\first\second\else\ifx\first\third\else[{##1}]\fi\fi}%
	}%
}
\def\Vardefx#1#2{%
	\expandafter\newcommand\csname #1\endcsname[1]{%
		\ensuremath{\mathsf{#2}[{##1}]}%
	}%
}

\Vardef{Rounds}
\Vardef{Count}
\Vardef{Tick}
\Vardef{Clock}
\Vardef{Synth}
\Vardef{Hybrid}
\Vardef{Coin}
\Vardef{Full}
\Vardef{Level}
\Vardefx{TickE}{tick_{E}}
\Vardef{Load}
\Vardef{Undecided}
\Vardef{Opinion}
\Vardef{PhaseType}
\Vardef{Round}
\Vardef{Clusterhead}
\Vardef{Leader}
\Vardef{Junta}
\Vardef{Phase}

\let\epsilon\varepsilon
\let\phi\varphi

\def\protocol#1{\textsc{#1}\xspace}

\AfterPackageLoaded{cleveref}{
    \crefname{lst@lineno}{Line}{Lines}
    \crefname{algoflt}{Algorithm}{Algorithms}
    \crefname{enumi}{Statement}{Statements}
}

\makeatletter
\def\?#1{}
\def\whp{w.h.p\@ifnextchar.{.\?}{\@ifnextchar,{.}{\@ifnextchar){.}{\@ifnextchar:{.:\?}{.\ }}}}}
\def\Whp{W.h.p\@ifnextchar.{.\?}{\@ifnextchar,{.}{.\ }}}
\makeatother

\let\E\Ex
\newcommand{\Probs}[1]{\Pr\left[{#1}\right]}

\newcommand{\xconfig}{\mathbf{x}}
\newcommand{\Xconfig}{\mathbf{X}}
\newcommand{\tconfig}{\tilde{\mathbf{x}}}
\newcommand{\Tconfig}{\tilde{\mathbf{X}}}
\newcommand{\better}{\succeq_Q}
\newcommand{\dominates}{\succeq_C}

\def\shBias{\xi\xspace}
\newcommand{\submConst}{r\xspace}
\newcommand{\UpperUBound}{\Phi_{up}}

\newcommand{\USD}[1]{\mathrm{USD}_{#1}}
\newcommand{\usd}[1]{\delta_{#1}}
\newcommand{\filter}{\mathcal{F}}
\newcommand{\productive}{\mathrm{prod}_t}

\newcommand{\sref}[1]{Statement \ref{#1}}

\newcommandx{\chr}[2][1=]{\todo[linecolor=orange,bordercolor=orange,backgroundcolor=orange!40,#1]{#2}}
\newcommandx{\felix}[2][1=]{\todo[linecolor=blue,bordercolor=blue,backgroundcolor=blue!30,#1]{#2}}
\newcommandx{\petra}[2][1=]{\todo[linecolor=red,bordercolor=red,backgroundcolor=red!40,#1]{#2}}

\title{Undecided State Dynamics with Stubborn Agents}
\subtitle{}

\author{Petra Berenbrink}{Universität Hamburg, Germany}{petra.berenbrink@uni-hamburg.de}{}{}

\author{Felix Biermeier}{Universität Hamburg, Germany}{felix.biermeier@uni-hamburg.de}{}{}

\author{Christopher Hahn}{Universität Hamburg, Germany}{tim.christopher.hahn@uni-hamburg.de}{}{}

\authorrunning{P.\ Berenbrink, F.\ Biermeier, C.\ Hahn}
\Copyright{Petra Berenbrink, Felix Biermeier, and Christopher Hahn}
%more specific?

\hideLIPIcs
\nolinenumbers
\begin{document}
\setcounter{page}{0}
\maketitle
\thispagestyle{empty}
\begin{abstract}
    
In the classical Approximate Majority problem with two opinions there are agents with Opinion 1 and with  Opinion 2.
The goal is to reach consensus and to agree on the majority opinion if the bias is sufficiently large.
It is well known that the problem can be solved efficiently using the \emph{Undecided State Dynamics} (USD) where an agent interacting with an agent of the opposite opinion becomes undecided. 
In this paper, we consider a variant of the USD with a preferred Opinion 1. That is, agents with Opinion 1 behave \emph{stubbornly} -- they preserve their opinion with probability $p$ whenever they interact with an agent having Opinion 2.
Our main result shows a phase transition around the stubbornness parameter $p \approx 1-x_1/x_2$.
If $x_1 = \Theta(n)$ and $p \geq 1-x_1/x_2 + o(1)$, then all agents agree on Opinion 1 after $O(n\cdot \log n)$ interactions.
On the other hand, for $p \leq 1-x_1/x_2 - o(1)$, all agents agree on Opinion 2, again after $O(n\cdot \log n)$ interactions.
Finally, if $p \approx 1-x_1/x_2$, then all agents do agree on one opinion after $O(n\cdot \log^2 n)$ interactions, but either of the two opinions can survive. All our results hold with high probability.

\end{abstract}

\clearpage

%\tableofcontents
%\listoftodos[TODOs]

\newpage

\section{Introduction}
In this paper we consider a consensus problem 
for population protocols. We have $n$ indistinguishable agents with limited computational power and memory.
The agents communicate via pairwise interactions. Here we assume that the pairs are chosen uniformly at random. 
The population model is widely applicable to real-world scenarios where the communication pattern between the agents is unpredictable.
Examples include sensor networks, mobile ad-hoc networks, swarm robotics and epidemiology.

The consensus problem studied here is defined as follows. Initially, each agent holds one of two possible opinions (called $1$ and $2$ in the following).
The goal of the agents is now to agree on an opinion (which has initially positive support). 
A special case of a consensus problem is the \emph{Majority} problem where the goal is to agree on the opinion which initially has the larger support.
One can distinguish between approximate majority where agents have to agree on the majority opinion only if there is a big difference in the support and exact majority where the agents have to agree on the majority opinion even if the support differs only by one. 
A well-known approximate majority protocol is the so-called \emph{undecided state dynamics}. Here the agents can adopt one additional state called the undecided state. 
Assume agent $q$ interacts with agent $q'$.
If both agents have a different Opinion, $q$ becomes undecided. 
If $q$ is undecided, it adopts the Opinion of $q'$.

In this paper we consider a biased variant of the undecided state dynamics.
We assume that $1$ is the preferred opinion.
Again let us assume agent $q$ interacts with agent $q'$.
If $q$ has the preferred Opinion $1$ and $q'$ has Opinion $2$, then agent $q$ becomes undecided with probability $1-p$.
With probability $p$ it does not change its opinion at all.
We can regard the agent as being stubborn; with probability $p$ the agent insists on its own opinion.
All the other interactions between $q$ and $q'$ are identical to the interactions in the undecided state dynamics.

Our model is a natural extension of the undecided state dynamics to biased agents.
We believe the model covers many real-world instances.
For example, referendums are often triggered by a small active minority.
A newly evolved disease of higher infectiousness can replace the older one. 
A more plausible rumor, or the use of a better tool or newer technology all start out as a minority.
Sometimes, we are interested in limiting the influence of a minority. 
For example, population protocols are examined for their \emph{robustness} against few erroneous or malicious agents. One might model this as a stubborn opinion and see at what threshold the faithful execution breaks down.

\paragraph{Our Contribution}
In this paper we are especially interested in scenarios where the initial majority opinion is $2$. 
Since $p>0$ gives the preferred opinion an advantage, we focus on initial configurations where the preferred opinion is the minority.
We show that for any configuration with linear support for both opinions there exists a threshold $p_s$ such that Opinion 1 wins after $O(n \log n)$ interactions if the stubbornness parameter $p$ is slightly larger than $p_s$. 
If $p$ is slightly smaller than $p_s$, Opinion $2$ wins after $O(n \log n)$ interactions. If $p \approx p_s$, we show that one of the two opinion wins in $O(n \log^2 n)$ interactions, but it is not clear which of the two opinions wins. Our results show that, even if the initial support for Opinion $2$ is larger, for sufficiently large $p$ the agents will still agree on  Opinion $1$. 
Additionally, we show a majorization results which allow us, given two initial configurations and bias-pairs, to determine which of the two systems has the smaller convergence time. 
The exact statement of our main results is given in \cref{sec:problem_definition_main_result}.

\subsection{Related Work}
\paragraph{Undecided State Dynamics}
The USD was independently introduced by Angluin et al.~\cite{DBLP:journals/dc/AngluinAE08} for the population protocol model and by Perron et al.~\cite{DBLP:conf/infocom/PerronVV09} for the closely related asynchronous gossip model.
The latter model works in parallel rounds and all agents are active at the same time.
Both papers show that, for $k=2$ opinions, \whp the convergence time of the process is $O(n \log n)$ interactions (respectively, $O(\log n)$ continuous time) assuming an initial bias of $\omega(\sqrt{n} \log n)$. 
Condon et al.\ \cite{DBLP:conf/dna/CondonHKM17} present an improved result for the two-opinion case in the population model.
They show that the process solves the approximate majority problem assuming an initial bias of $\Omega(\sqrt{n \log n})$.
Their analysis is based on the more general chemical reaction network model and shows this result for a group of closely related processes.

Amir et al.~\cite{DBLP:conf/podc/AmirABBHKL23} analyze the $k$-opinion USD under mild constraints on $k$ and the initial number of undecided agents.
They show the process converges in $O(k \cdot n\log n)$ interactions regardless of the initial bias.
Additionally, they provide a more technical statement where the convergence time depends on the initially largest opinion.

In the parallel gossip model, the convergence of the $k$-opinion USD for 
$k \ge 2$ was first studied by Becchetti et al.~\cite{DBLP:conf/soda/BecchettiCNPS15}. 
To measure the convergence time of an initial configuration, they introduce the 
\emph{monochromatic distance}.
Roughly speaking, this distance is the sum of squares 
of the support of each opinion, normalized by the square of the 
 largest opinion. They show convergence
within $O(\text{md}(\textbf{x}) \cdot \log n)$ parallel rounds assuming an initial multiplicative bias, where 
$\text{md}(\textbf{x})$ is the monochromatic distance of the initial
configuration, which is always bounded above by $k$.
In the two-color case, Clementi et al.\ \cite{DBLP:conf/mfcs/ClementiGGNPS18} later present a tight analysis (giving convergence rates that hold for any 
initial configuration) without using the monochromatic
distance.

Furthermore, in a related line of research, multiple works \cite{DBLP:conf/podc/GhaffariP16a,DBLP:conf/icalp/BerenbrinkFGK16,DBLP:conf/podc/BankhamerEKK20,DBLP:conf/soda/BankhamerBBEHKK22} have analyzed a \emph{synchronized} variant of the USD where the system alternates between two different phases in a synchronized fashion.
In the first phase, all agents perform one step of the USD.
In the second phase, all undecided agents adopt an opinion again.
This provides a poly-logarithmic convergence time regardless of the initial configuration.
The downside of this variant is a significant state overhead by the use of so-called \emph{phase clocks}.
The \emph{phase clocks} allow the synchronization of the agents.

The authors in \cite{DBLP:journals/swarm/DAmoreCN22} study the $2$-USD with uniform noise in a synchronous time model (gossip). Whenever an agent communicates with another agent, it observes its actual state only with probability $1-p$.
Otherwise, it observes any state uniformly at random.
The main result is a phase transition regarding the probability $p$.
When the probability $p$ is less than 1/6, the configuration quickly reaches a meta-stable almost consensus. On the other hand, when the probability is greater than 1/6, the initial majority is lost in $O(\log n)$ rounds. 

\paragraph{Other Opinion Dynamics}

\def\protocol{}

There is a popular family of so-called $j$-Majority dynamics.
Every agent samples randomly $j$ other agents and adopts the majority opinion among the sample.
The special case $j$=1 is known as \protocol{Voter} dynamics \cite{DBLP:journals/iandc/HassinP01, DBLP:journals/networks/NakataIY00, DBLP:conf/podc/CooperEOR12, DBLP:conf/icalp/BerenbrinkGKM16, DBLP:conf/soda/KanadeMS19}.
The variants for $j$=2 and $j$=3 have been analyzed under the names of \protocol{Two-Choices} dynamics \cite{DBLP:conf/icalp/CooperER14,DBLP:conf/wdag/CooperERRS15,DBLP:conf/wdag/CooperRRS17} and the \protocol{3-Majority} dynamics \cite{DBLP:journals/dc/BecchettiCNPST17, DBLP:conf/podc/GhaffariL18, DBLP:conf/podc/BerenbrinkCEKMN17}.
For further references, we refer the reader to the survey of consensus dynamics by Becchetti et al.\ \cite{DBLP:journals/sigact/BecchettiCN20}.

D'Amore and Ziccardi~\cite{DBLP:conf/sirocco/DAmoreZ22} consider uniform communication noise for the $3$-Majority dynamics. They observe a similar phase transition as in \cite{DBLP:journals/swarm/DAmoreCN22}.

In \cite{mobilia2003does}, Mobilia examined the role of a single so-called \emph{zealot} -- an agent that never changes its opinion -- for the Voter dynamics.
Mobilia et al.~\cite{mobilia2007role} pursued this further for several zealots.
Yildiz et al.~\cite{yildiz2013binary} examined the role of two sets of zealots with opposing opinions -- which they named \emph{stubborn agents}.

In \cite{DBLP:conf/ijcai/BecchettiCKPTV23}, Becchetti et al. consider a constant number of opinions and a single stubborn agent.
In each time step, an agent is activated uniformly at random and samples $\ell$ opinions of other agents uniformly at random.
They show for a constant number of initial opinions that every memoryless dynamics requires $\Omega(n^2)$ time steps in expectation to converge.
\paragraph{Biased Opinions}
All results for biased opinion dynamics consider the case with two opinions. One of the opinions is the preferred one. 
Anagnostopoulos et al.~\cite{DBLP:journals/isci/Anagnostopoulos22} consider biased opinion dynamics consisting of two steps. 
In each time step, an agent is selected uniformly at random and adopts the preferred opinion with probability $\alpha$. Otherwise, the agent adopts the majority among its neighbors' opinions.
Note that in this setting an agent can adopt the preferred opinion even if none of its neighbors shares that opinion. Hence, in contrast to our model, the system has only one absorbing state where all agents agree on the preferred opinion. The authors show that there is a phase transition for $\alpha = 1/2$ in dense graphs for the majority rule. For the process with voter rule (the node adopts a random neighbor's opinion) they do not observe any phase transition; the absorption time is $O(n\log n)$.

Mukhopadhyay et al.~\cite{mukhopadhyay_2023} consider the same biased opinion variant as \cite{DBLP:journals/isci/Anagnostopoulos22} but with the 2-choices rule.
For the complete graph, they bound the expected absorption time and observe a phase transition around $\alpha = 1/9$. 
For large $\alpha > 1/9$, the process converges in time $O(n\log n)$ to the preferred opinion.
On the other hand, for small $\alpha < 1/9$, the convergence time depends on the initial fraction of the preferred opinion among the population. If the initial support of the preferred opinion is sufficiently large, the process converges fast ($O(n\log n)$), and for a small initial support the process needs at least $\Omega(\exp(n))$ steps.

Cruciani et al.~\cite{DBLP:journals/dc/CrucianiNNS21} consider a variant of the biased opinion dynamics on core-periphery networks.
The network consists of a core that is a densely-connected subset of agents with the same opinion.
The remaining agents form the periphery with another opinion.
They observe a phase transition that depends on the cut between the core and periphery.
Either all agents agree relatively fast on the initial opinion of the core agents.
Otherwise, the process remains in a meta-stable where both opinions remain in the network for at least polynomial many rounds.

Cruciani et al.~\cite{DBLP:journals/dc/CrucianiMQR23} study the $j$-majority dynamics. In each time step, each agent simultaneously samples $j$ neighbors uniform at random and it adopts the majority opinion.
They consider two different noise models that alter the communication between agents with probability $p$.
In the first variant an agent may observe the preferred opinion instead of the sampled opinion.
In the second variant the opinion of an agent may directly change to be the preferred one.
In both variants, they show phase transitions for some $p^*$ with the preferred opinion being the initial minority.
Eventually, all agents agree on the preferred opinion \whp.
For small $p < p^*$ and $j \geq 3$, it requires $n^{\omega(1)}$ parallel rounds.
For large $p > p^*$ and $j \geq 3$, it only requires $O(1)$ parallel rounds.
At last for $j <3$, it requires $O(1)$ parallel rounds for every $p>0$.

\section{Problem Definition \& Main Result}
\label{sec:problem_definition_main_result}
In this paper, we consider the \emph{population protocol} model with a set of $n$ anonymous agents. 
In each time step $t$, a random scheduler picks an ordered pair of agents $(i,j)$ uniformly at random.  We call agent $i$ the \emph{initiator} and agent $j$ the \emph{responder}. A function $\delta: Q^2 \rightarrow Q$ defines the state transitions. We assume that only initiators change their state according to $\delta$, i.e.,  $q_i(t+1)=\delta(q_i(t),q_j(t))$ and $q_j(t+1)=q_j(t)$.

We assume that  $Q=\set{1,2,\bot}$ and we say  that agent $i$ has \emph{Opinion} $1$ ($2$) if $q_i = 1$ ($q_i = 2$). Otherwise, if $q_i = \bot$, we say that agent $i$ is \emph{undecided}. 
We assume that Opinion $1$ is the preferred opinion. The random variable $X_i(t)$ denotes the \emph{support} of Opinion $i$ at time $t$, i.e., the number of agents in state $i$. The random variable  $U(t)$
denotes the number of undecided agents at time $t$.
We call a vector $\mathbf{X}(t) = (X_1(t), X_2(t), U(t)) \in \N^{3}$ with $X_1(t)+X_2(t)+U(t) = n$ a \emph{configuration}. In the following, we use capital letters for random variables and lowercase letters ($(x_1(t), x_2(t), u(t))$ for fixed outcomes.
We use $\mathcal{F}_t = (F_i)_{i=0}^{t}$ to denote the natural filtration consisting of the initial configuration at time $0$ and all random choices up to time $t$: the interacting agents, and the outcome of the interaction based on the stubbornness $p$.
I.e., We write $\mathcal{F}_t$ for $\Xconfig(0) = \xconfig(0), \Xconfig(1) = \xconfig(1), \ldots, \Xconfig(t) = \xconfig(t)$ and for the sake of readability we may use $\xconfig$ instead of $\xconfig(t)$.

\paragraph{Undecided State Dynamics with Stubborn Agents}
The undecided state dynamics (USD) is defined as follows. Whenever an agent encounters another agent with a different opinion, it drops its opinion and becomes undecided. Whenever an undecided agent encounters another agent, it adopts its opinion.

The following generalization of the USD has a \emph{preferred} opinion, which is w.l.o.g.\ Opinion $1$. We call the other opinion \emph{unpreferred}.
To model the preference on Opinion $1$, we assume that agents with Opinion $1$ are \emph{stubborn} in the following sense. Whenever an agent with Opinion $1$ meets an agent with Opinion $2$ it does not become undecided immediately. 
Instead, it draws a random number in $\intoc{0,1}$ and keeps its opinion if this number is smaller or equal to a constant $p \in \intcc{0,1}$. Otherwise it becomes undecided.
In the following, we call $p$ the \emph{stubborness}.
In more detail, if the random scheduler picks a pair of agents with states $(1,2)$, the initiator remains unchanged with probability $p$.
With probability $1-p$, its new state is $\delta(1,2) = \bot$.
All other interactions remain as in the original version of the undecided state dynamics. 

Formally, the transition function of the stubborn USD with stubbornness $p$ is
\begin{align*}
    %\delta_p: Q\times Q \rightarrow Q:
    (q,q') & \mapsto 
    \begin{cases}
    \bot \text{ if } q=2, q'=1\\
    \bot \text{ if } q=1, q'=2 \text{ with probability } 1-p\\
    q' \text{ if } q = \bot
    %\land q' \neq \bot
    \\
    q \text{ otherwise.}
    \end{cases}
\end{align*}%
In the following, we refer to the stubborn USD process with stubbornness $p$ by $\USD{p}$ and to its transition function by $\usd{p}$. $\USD{p}(\xconfig)$ is defined as the $\USD{p}$
with initial configuration $\xconfig$. 
Consequently, the standard USD process is $\USD{0}$ and its transition function $\usd{0}$.

\subsection{Results}
Our focus is on the convergence time of $\USD{p}(\xconfig)$, defined as the number of interactions until all agents agree on one of the two opinions if $\USD{p}$ is started on initial configuration $\xconfig = (x_1,x_2,u)$.
Let $T_i(p,\xconfig)$ be the convergence time of process $\USD{p}(\xconfig)$ assuming  opinion $i$ survives.
In \cref{thm:main-theorem}, we show that there exists a phase transition around the \emph{threshold probability} $p_s = 1-x_1/x_2$.
If $p$ is sufficiently larger than $p_s$, the process will reverse the initial bias with high probability and the agents will agree on Opinion $1$.
For $p$ sufficiently smaller than $p_s$, with high probability all agents will agree on Opinion 2 instead.
In the intermediate cases any of the two opinions might win, but we can still provide bounds on the convergence time.
Note that $\Omega(n\log n)$ is a trivial lower bound for the process since this is the time until each agent is, w.h.p., activated at least once.

\begin{theorem}
\label{thm:main-theorem}
    Let $\epsilon, p \in \intoc{0,1}$ be arbitrary constants and let $\xconfig = (x_1,x_2,u)$ be a configuration with $x_1 \in \intcc{\epsilon \cdot n,x_2}, u \leq \frac{n}{2}$.
    Let $p_s \coloneqq 1-x_1/x_2$.
    Then the following statements hold \whp.
    \begin{align}
        T_1(p,\xconfig) &= O(n \cdot \log n) &&\text{if } p-p_s = \Omega\left(\sqrt{n^{-1}\cdot \log n}\right), \label{thm:main-theorem:1}\\
        T_2(p,\xconfig) &= O(n \cdot \log n) &&\text{if } p_s-p = \Omega\left(\sqrt{n^{-1}\cdot \log n}\right),\label{thm:main-theorem:2}\\
        T_{1\lor2}(p,\xconfig) &= O(n \cdot \log^2 n) &&\text{otherwise.}\label{thm:main-theorem:3}
    \end{align}
\end{theorem}

Note that for initial configurations $\xconfig = (x_1,x_2,0)$ with $x_1(0) > x_2(0)$, it is known that Opinion 1 is more likely to win in $O(n\log n)$ interactions, even for $p=0$ (see \cite{DBLP:conf/dna/CondonHKM17}).
We could easily rephrase \cref{thm:main-theorem} to cover those configurations.
However, in \cref{thm:coupling-theorem} we give a more general extension that might be of independent interest.
We will explain below which results follow from both Theorems. 

\begin{restatable}{theorem}{couplingtheorem}
\label{thm:coupling-theorem}
    Let $p,\tilde{p} \in \intcc{0,1}$ and let $\xconfig = (x_1,x_2,u), \tconfig = (\tilde{x}_1,\tilde{x}_2,\tilde{u})$ be arbitrary configurations. Then, it holds for all $t \geq 0$ that
    \begin{align}
        \Pr[T_1(p,\xconfig) \leq t] &\geq \Pr[T_1(\tilde{p},\tconfig) \leq t] &&\text{if } \tilde{p} \leq p, \tilde{x}_1 \leq x_1 \text{ and } \tilde{x}_2 \geq x_2,\label{thm:coupling-theorem:1}\\
        \Pr[T_2(p,\xconfig) \leq t] &\geq \Pr[T_2(\tilde{p},\tconfig) \leq t] &&\text{if } \tilde{p} \geq p, \tilde{x}_1 \geq x_1 \text{ and } \tilde{x}_2 \leq x_2.\label{thm:coupling-theorem:2}
    \end{align}
\end{restatable}

\Cref{thm:coupling-theorem} allows to transfer high-probability results in several ways. 
The aforementioned example $\USD{0}(\tconfig)$ with $\tilde{x}_1-\tilde{x}_2=\Omega(\sqrt{n \cdot \log n})$ converges \whp on Opinion 1 (\cite{DBLP:conf/dna/CondonHKM17}).
By \sref{thm:coupling-theorem:1}, e.g., $\USD{0.1}(\tconfig)$ and $\USD{0}(\xconfig)$ with $x_1 - x_2 = n/2$ also converge on Opinion 1 with at least the same probability.
Secondly, in $\USD{0.25}((n/5,n-n/5,0))$ all agents agree \whp on Opinion 2 by \cref{thm:main-theorem}.
By \sref{thm:coupling-theorem:2}, the same holds for $\USD{0.25}(\log n, n-\log n,0)$ and $\USD{1/n}(n/10,n-n/10,0)$.
We visualize our results in \cref{fig:proof-overview}.
\begin{figure}[H]
    \includegraphics[width=0.56\textwidth]{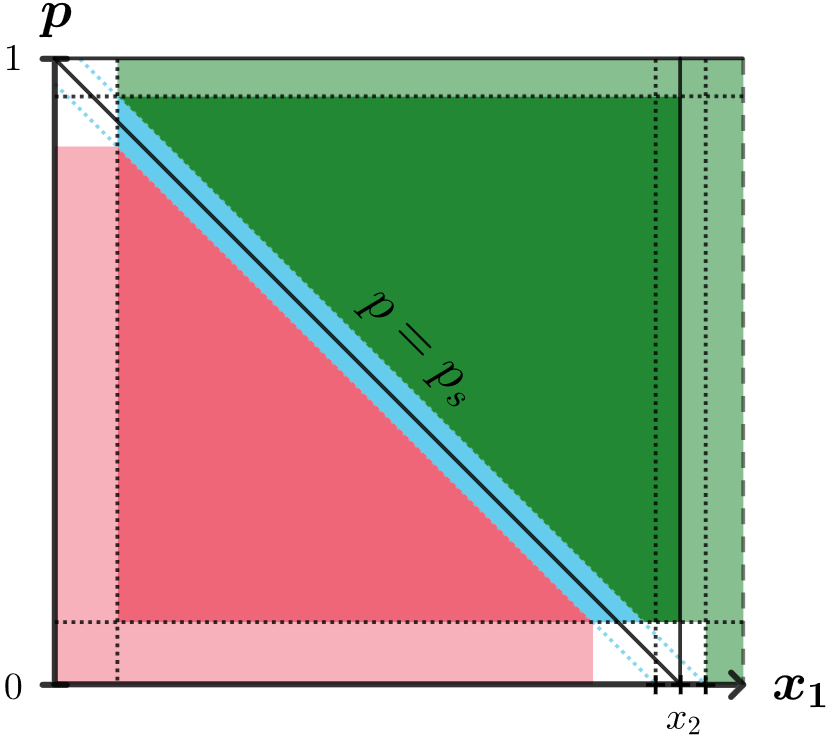}
    \caption{Overview over the results (not to scale). 
    The x-coordinate indicates $x_1$ from $1$ to $x_2 + o(n)$ -- we omit larger values of $x_1$.
    The y-coordinate indicates the stubbornness $p$ from $0$ to $1$.
    The black diagonal line represents $p=p_s$ and the dashed blue lines around it represent $p = p_s \pm \Theta(\sqrt{n^{-1} \cdot \log n})$. At $p=0$, they coincide with the known bounds from \cite{DBLP:conf/dna/CondonHKM17}: Opinion 1 wins \whp for $x_1 \geq x_2 + \Omega(\sqrt{n\log n})$, Opinion 2 wins \whp for $x_1 \leq x_2 - \Omega(\sqrt{n\log n})$.
    The inner opaque rectangular area corresponds to \cref{thm:main-theorem}.
    The remaining colored areas correspond to \cref{thm:coupling-theorem}.
    The area colored in green represents where Opinion 1 wins \whp in $O(n\log n)$ interactions.
    The area colored in red represents where Opinion 2 wins \whp in $O(n\log n)$ interactions.
    In the blue area, either Opinion wins \whp in $O(n\log^2n)$ interactions.
    \label{fig:proof-overview}
    }
\end{figure}

\paragraph{Main Idea of the Analysis}
The case $p=0$ is well-studied. A standard approach for this case is to track the evolution of the support of the opinions, e.g., using the (additive) \emph{bias} $x_1(t)-x_2(t)$ or the (multiplicative) \emph{gap} $x_1(t)/x_2(t)$ as a potential function (see \cite{DBLP:journals/dc/AngluinAE08,DBLP:conf/dna/CondonHKM17,DBLP:conf/podc/AmirABBHKL23}). 
As a first contribution, we generalized this approach by identifying the  driving force of the biased version of the undecided state dynamics which  we call \emph{weighted bias} (as a function of $t$) $\Delta_w(t)$.
The weighted bias at time $t$ is defined as $x_1(t)-(1-p)x_2(t)$.
The initial weighted bias plays a similar role in determining the winning opinion as the initial bias does in the classical USD with $p = 0$.
Note that an \emph{initial weighted bias} $\Delta_w(0)$ of $c \cdot x_2(0)$ is equivalent to $p = 1- x_1(0)/x_2(0) + c = p_s+c$. 
The weighted bias is a different way of looking at the problem that is more practical.

The most difficult part of our analysis are those configurations where the stubbornness parameter balances out the initial support deficit of the preferred opinion, i.e., $p=p_s$ and equivalently $\Delta_w(0) = 0$.
In contrast to the corresponding case with $p=0$ and $x_1 = x_2$, where the initial bias is zero, creating a sufficiently large weighted bias is more involved.
The standard way is to define a random walk on the integers and apply known anti-concentration and concentration bounds.
This approach is not viable here, since $|\Delta_w(t+1)-\Delta_w(t)| \in \set{0,1-p,1}$ leads to non-integer states.
Instead, we exploit the submartingale property of $\Delta_w(t)$ itself and the function $Y_t = \Delta_w^2(t) - r\cdot t$ for a suitably chosen value of $r$.

\paragraph{Outline of the paper}
The remainder of this paper is organized as follows.
We prove each Statement of \cref{thm:main-theorem} into its own section, but we focus on Statement 1 and 3 (\cref{sec:analysis:case1,sec:analysis:case3}).
The proof of \sref{thm:main-theorem:2} is almost identical to \sref{thm:main-theorem:1}, except for minor details in some calculations. It can be found in \cref{sec:analysis:case2}.
In \cref{sec:coupling}, we prove \cref{thm:coupling-theorem}.
We give a conclusion in \cref{sec:conclusion}.
Auxiliary results can be can be looked up in \cref{apx:auxiliary-results}.
For $|p-p_s| = o(1)$, we need some auxiliary bounds on the number of undecided agents that we put into \cref{sec:undecided}.
For the sake of completeness, we analyze $p=1$ separately in \cref{ssec:p=1}.

\section{Cases in which Opinion 1 wins}
\label{sec:analysis:case1}
In this section, we show \sref{thm:main-theorem:1} of \cref{thm:main-theorem}, namely that Opinion 1 wins if the stubbornness $p$ is sufficiently larger than $1-x_1(0)/x_2(0)$.
We show in \cref{lem:drift_x2_to_zero} that the gap increases whp.\ such that there exists a time $T_1 = O(n\log n)$ where each agent either has Opinion $1$ or is undecided.
From there on, it is easy to show that all agents agree on Opinion $1$ after an additional $O(n \log n)$ steps.
As an auxiliary result for \cref{lem:drift_x2_to_zero}, we first show in \cref{lem:shifted_bias_bound:case1} that the weighted bias does not halve during the first interactions.

\begin{lemma}
\label{lem:shifted_bias_bound:case1}
    Let $\textbf{x}(t_0)$ be a configuration with weighted bias $\Delta_w(t_0) > 0$.
    Let $\xi(\tau)$ be the event that  $\Delta_w(t) \geq \Delta_w(t_0)/2$ for all $t\in \intcc{t_0,\ldots, t_0+\tau}$.
    Then, with probability at least $1-n^{-6}$, $\xi(T)$ holds for all $ \tau \leq \Delta_w^2(t_0)/(16 \ln n)$.
\end{lemma}
\begin{proof}
    We aim to apply an Azuma-Hoeffding-bound (\cref{lem:azuma-hoeffding}) to $\Delta_w(t)$ for each $t \in \intcc{t_0,t_0+\tau} $ with $ \tau = \Delta_w^2(t_0)/(16 \ln n)$.
    Fix an arbitrary $t \leq \tau$.
    To apply \cref{lem:azuma-hoeffding}, we need to show that $\Ex{\Delta_w(t+1) - \Delta_w(t) | \mathcal{F}_{t}} \geq 0$ and that $|\Delta_w(t+1)-\Delta_w(t)|$ is bounded.
    
    First, we calculate the expected change in $\Delta_w$ by considering all possible interactions. For ease of presentation, we drop the parameter $t$ whenever clear from the context.
    With probability $x_1 \cdot x_2/n^2$, the randomly chosen initiator has Opinion $1$, and the responder has Opinion $2$. In that case, the initiator is stubborn (it does not change its state) with probability $p$, resulting in $\Delta_w(t+1) = \Delta_w(t)$. With probability $1-p$ the initiator becomes undecided and $\Delta_w(t+1) = \Delta_w(t) -1$.
    With probability $x_2 \cdot x_1/n^2$, the initiator loses its Opinion $2$ and becomes undecided, resulting in $\Delta_w(t+1) = \Delta_w(t)+(1-p)$.
    Whenever an undecided agent initiates an interaction where the responder has either Opinion $1$ or Opinion $2$, it adopts the responder's opinion.
    Such interactions occur with probability $u \cdot x_1/n^2$ and $u \cdot x_2/n^2$, resulting in $\Delta_w(t+1) = \Delta_w(t) +1$ and $\Delta_w(t+1) = \Delta_w(t)-(1-p)$.
    At last, there exist neutral interactions that do not change the potential, i.e., $\Delta_w(t+1) = \Delta_w(t)$.
    These interactions occur with the remaining probability $(x_1^2+x_2^2+n\cdot u)/n^2$.
    \begin{align*}
        \MoveEqLeft\Ex{\Delta_w(t+1)-\Delta_w(t) | \mathcal{F}_{t}}  \\
        & = \frac{x_1 \cdot x_2}{n^2} \cdot \left( p\cdot0 + (1-p)\cdot(-1)\right)
        +   \frac{x_2 \cdot x_1}{n^2} \cdot (1-p)
        +   \frac{u \cdot x_1}{n^2} \cdot(+1)\\
        &\phantom{={}}+   \frac{u \cdot x_2}{n^2} \cdot (-1+p)
        +   \frac{x_1^2 + x_2^2 + n \cdot u }{n^2}\cdot 0 \\
        & = -\frac{(1-p)\cdot x_1\cdot x_2}{n^2} + \frac{(1-p)\cdot x_2\cdot x_1}{n^2} + \frac{u\cdot x_1}{n^2} - \frac{(1-p)\cdot u\cdot x_2}{n^2} \\
        & = \frac{u}{n^2} \cdot \Delta_w(t) \geq 0.
    \end{align*}    
    
    It remains to show that $|\Delta_w(t+1)- \Delta_w(t)|$ is bounded. Here, we show that the term is bounded by $1$.
    Consider every possible interaction at time $t+1$.
    If the initiator does not change its opinion, $|\Delta_w(t+1) - \Delta_w(t)| = 0$.
    If the support of Opinion $1$ changes, we have $|\Delta_w(t+1) - \Delta_w(t)| = 1$.
    Otherwise, the support of Opinion $2$ changes by one and $|\Delta_w(t+1) - \Delta_w(t)| = 1-p \leq 1$.

    Now we are ready to apply \cref{lem:azuma-hoeffding} with $\lambda = \Delta_w(t_0)/2$.
    \begin{align*}
        \MoveEqLeft \Probs{\left. \Delta_w(t) < \frac{\Delta_w(t_0)}{2} ~\right|~ \mathcal{F}_{t_0}} 
        = \Probs{\Delta_w(t) - \Delta_w(t_0) < -\lambda ~|~ \mathcal{F}_{t_0}} 
        \leq \exp\left(-\frac{2\lambda^2}{t}\right)\\
        &\leq \exp\left(-\frac{2\lambda^2}{\tau}\right)
        \leq \exp\left(-8 \ln n\right).
    \end{align*}
    By application of the union bound over the first $\tau$ interactions, the statement holds 
    with probability of at least
    $
        1- \tau \cdot \exp(-8 \ln n) 
        \geq 1-n^2 \cdot n^{-8}
        = 1-n^{-6}.
    $
\end{proof}

Note that \cref{lem:shifted_bias_bound:case1} is an auxiliary result since it only shows that the weighted bias is not decreasing too much for $\Omega(\Delta_w^2(t_0)/\log n)$ time.
We use this in the following lemma to show that the support of the initial majority opinion drops to zero during that time.

\begin{lemma}
\label{lem:drift_x2_to_zero}
     Let $\textbf{x}(t_0)$ be a configuration with weighted bias $\Delta_w(t_0) \geq c_s\cdot n$ for an arbitrary constant $c_s$.
    Let $T_1 = \inf\{t\geq 0 ~|~ x_2(t) = 0 \}$.
    Then, $ \Probs{T_1 \leq 20 \cdot c_s^{-1} \cdot n\log n} \geq 1-n^{-2}.$
\end{lemma}

\begin{proof}
    Let $\Psi(t) = x_2(t)/x_1(t)$ denote the inverse of the \emph{gap}.
    The idea is to show that this potential function decreases exponentially and apply a known drift theorem.
    Similar to the proof of \cref{lem:shifted_bias_bound:case1}, we calculate the expected change in $\Psi$ by considering all possible interactions.
    
    \begin{align*}
        \MoveEqLeft\E{\Psi(t+1)-\Psi(t) ~|~ \filter_t} \\
        &=  \frac{x_1 \cdot x_2}{n^2}\left( p\cdot 0 + (1-p)\cdot \left(\frac{x_2}{x_1-1}-\frac{x_2}{x_1} \right) \right)
        +   \frac{x_2 \cdot x_1}{n^2} \left(\frac{x_2-1}{x_1} - \frac{x_2}{x_1} \right)\\
        &\phantom{={}} + \frac{u \cdot x_1}{n^2} \left(\frac{x_2}{x_1+1} - \frac{x_2}{x_1}\right)
        + \frac{u \cdot x_2}{n^2} \left(\frac{x_2+1}{x_1} - \frac{x_2}{x_1}\right)
        +  \frac{x_1^2 + x_2^2 + n\cdot u}{n^2} \cdot 0\\
        &=  \frac{(1-p)x_1 \cdot x_2^2}{n^2}\left( \frac{1}{x_1(x_1-1)} \right)
        -   \frac{x_2 \cdot x_1}{n^2 \cdot x_1}
        -   \frac{u \cdot x_1 \cdot x_2}{n^2 \cdot x_1(x_1+1)}
        +   \frac{u \cdot x_2}{n^2 \cdot x_1}\\
        &= \frac{\Psi}{n^2}\cdot \left( \frac{(1-p)x_2}{x_1-1} - x_1 - \frac{u \cdot x_1}{x_1+1} + u \right)\\
        &= - \frac{\Psi}{n^2} \cdot \left( x_1 - (1-p)x_2 + \frac{(1-p)x_2}{x_1-1} - \frac{u}{x_1+1} \right).
    \end{align*}
    Observe that the expected potential change in $\Psi$ is a function of the weighted bias.
    We bound the weighted bias using \cref{lem:shifted_bias_bound:case1}:
    $\Delta_w(t) \geq \Delta_w(t_0)/2$ with probability at least $1-n^{-6}$ for all $t \in \intcc{t_0,t_0+c_3n\log n} $.
    We trivially bound $(1-p)x_2/(x_1-1) \geq 0$.
    To bound the term $u/(x_1+1)$, we use $x_1 \geq \Delta_w(t) \geq c_s\cdot n/2$ and $u \leq n$.
    Then, it holds \whp that
    \begin{align*}
        \MoveEqLeft\E{\Psi(t+1)-\Psi(t) ~|~ \filter_t}
        = - \frac{\Psi}{n^2} \cdot \left( \Delta_w(t) + \frac{(1-p)x_2}{x_1-1} - \frac{u}{x_1+1} \right)\\
        &\leq - \frac{\Psi}{n^2} \cdot \left( \Delta_w(t) - \frac{u}{x_1+1} \right) \\
        &\leq - \frac{\Psi}{n^2} \cdot \left(  \frac{c_s\cdot n}{2}  - \frac{n}{(c_s/2)\cdot n +1} \right) \\
        &\leq - \frac{\Psi}{n} \cdot \left(  \frac{c_s}{2}  - \frac{2}{c_s\cdot n +2} \right) \\
        &\leq -\frac{\Psi}{n} \cdot \frac{c_s}{4}.
    \end{align*}
    We now apply the multiplicative drift theorem (\cref{thm:mult_drift_tail_lengler_18}, found in the appendix) to bound $T_1$ with $r = 3 \ln n$, $s_0 =x_2(t_0)/x_1(t_0) \leq n$, $s_{\min} = (n-1)^{-1}$ and $\delta = n^{-1} \cdot c_s/4 $. 
    Then, we get
    \begin{align*}
        \Probs{T_1 > \frac{20 n\cdot \ln n}{c_s}} 
        &\leq \Probs{ T_1 > \left\lceil \frac{r+\ln(s_0/s_{\min})}{\delta} \right\rceil }
        \leq e^{-r}
        = n^{-3}.
    \end{align*}
    Note that in order to apply \cref{thm:mult_drift_tail_lengler_18}, we have to have $\E{\Psi(t)-\Psi(t+1) ~|~ \filter_t} \geq \delta \cdot \Psi(t)$ for all $\Psi(t) \neq 0$ and all $t\geq t_0$.
    \cref{lem:shifted_bias_bound:case1} asserts this only with high probability and for a limited time.
    But we can consider a process that deterministically jumps to configuration $(n,0,0)$ at time $t+1$ if for any $t$ \cref{lem:shifted_bias_bound:case1} is violated.
    We can apply \cref{thm:mult_drift_tail_lengler_18} to this process, and our original process behaves identical to it with probability $1-n^{-6}$.
    Thus, the lemma follows from the union bound.   
\end{proof}

Recall that by assumption of \sref{thm:main-theorem:1} \cref{thm:main-theorem}, we have $\Delta_w(0) = \Omega(\sqrt{n}\cdot \log n)$.
We now show that the the weighted bias doubles every $\BigO{n}$ interactions until it is of size $\Theta(n)$ -- the concrete bound of $n/10$ that we show was chosen rather arbitrarily.
Then we can apply \cref{lem:drift_x2_to_zero}.
The proof of \cref{lem:repeated-doubling} requires bounds on the number of undecided agents (\cref{lem:undecided_u_bound} and \cref{lem:lower_bound_u}) that we postpone to \cref{sec:undecided}.

\begin{lemma}
    \label{lem:repeated-doubling}
    Let $\xconfig(t_0)$ be a configuration with $\Delta_w(t_0) \geq \shBias \cdot \sqrt{n\log n}$ and let $T = \inf\set{t\geq t_0 ~|~ \Delta_w(t) \geq  n/10}$.
    Then $\Probs{T \leq (\shBias^2/6)\cdot  n\log n}\geq 1-n^{-3}$. 
\end{lemma}
\begin{proof}
    The proof idea is inspired by a repetition of the Gambler's ruin problem. 
    For $1\leq \ell \leq \log n$, we define the time intervals $\mathcal{I}_{\ell} =  \set{T_{\ell-1},\ldots,T_{\ell}-1}$ with $T_0 = t_0$ and
    \begin{align*}
        T_{\ell} \coloneqq \inf\set{t\geq t_0 ~|~ \Delta_w(t) \geq \min(2^\ell \cdot \Delta_w(t_0),(n-u(t))/4)}.
    \end{align*}%
    We show that the weighted bias leaves the interval $\intcc{\Delta_w(T_{\ell})/2,2\Delta_w(T_{\ell})}$ within $O(n)$ interactions, and that the value $2\Delta_w(T_{\ell})$ is reached before $\Delta_w(T_{\ell})/2$. 
    We apply this result repeatedly until $\Delta_w(T_{\ell}) \geq (n-u(t))/4$.
    Recall from \cref{lem:shifted_bias_bound:case1} that $(\Delta_w(t))_{t\geq t_0}$ is a submartingale with
    \begin{align*}
        \Ex{\Delta_w(t+1) ~|~\filter_{t}}
        \geq \Delta_w(t) + \frac{u(t) \cdot \Delta_w(t)}{n^2}
        \geq \Delta_w(t).
    \end{align*}
    
    Assume that $\Delta_w(T_{\ell}) \geq \shBias\cdot \sqrt{n\log n}$ and 
    let $T_{\ell,\min} \coloneqq \inf\set{t\geq T_{\ell} ~|~ \Delta_w(t) < \Delta_w(T_{\ell})/2}$ and $\tau \coloneqq (\shBias^2/6) \cdot n$.
    Then we bound $\Probs{T_{\ell,\min} > \tau}$ by using the Azuma-Hoeffding bound (\cref{lem:azuma-hoeffding}) with $\lambda = \Delta_w(T_{\ell})/2$:
    \begin{align*}
        \MoveEqLeft\Probs{T_{\ell,\min} < \tau} 
        =\Probs{\Delta_w(T_{\ell}+\tau) < \Delta_w(T_{\ell}) / 2 }
        = \Probs{\Delta_w(T_{\ell}+ \tau) - \Delta_w(T_{\ell}) < -\lambda} \\
        &\leq e^{-\frac{2\lambda^2}{(\shBias^2 / 6)n}}
        \leq e^{-\frac{3\shBias^2 n\log n}{\shBias^2n}}
        \leq n^{-3}.
    \end{align*}
        
    Next, let $T_{\ell,\max} \coloneqq \inf\set{t\geq T_{\ell} ~|~ \Delta_w(t) \geq  2\Delta_w(T_{\ell})}$.
    We bound $\Probs{T_{\ell,\max} > \tau}$ assuming $\Delta_w(t) \geq \Delta_w(T_{\ell})/2$ for all $t \in [T_{\ell}, T_{\ell}+\tau]$.
    Let 
    \begin{align*}
    Z \coloneqq \sum_{t =T_{\ell}+1}^{T_{\ell}+\tau} \Delta_w(t) \text{ and }
    & \mu \coloneqq \sum_{t =T_{\ell}+1}^{T_{\ell}+\tau}\Ex{\Delta_w(t) ~|~\filter_{t-1}}.
    \end{align*}
    In contrast to \cref{lem:shifted_bias_bound:case1}, to establish a sufficiently strong bias in the right direction we have to exploit the additional positive term $u(t) \cdot \Delta_w(t)/n^2$.
    
    From \cref{lem:lower_bound_u} it follows that \whp $u(t) \geq c_u\cdot n$  for constant $c_u = (3/10) \cdot (1-p)/(2-p)$ as long as $\Delta_w(t) < (n-u(t))/4$ holds.
    Note that this requires additional $O(n)$ interactions in the case $u(t_0)$ is too small.
    \Cref{lem:shifted_bias_bound:case1} guarantees that the weighted bias does not lose too much support during this time.
    Therefore, as long as $\Delta_w(t) \geq \Delta_w(T_{\ell})/2$ holds, we have
    \begin{align*}
        \Ex{\Delta_w(t+1) ~|~\filter_{t}}
        \geq \Delta_w(t) + \frac{u(t)\cdot \Delta_w(t)}{n^2}
        \geq \Delta_w(t) + \frac{c_u \cdot \Delta_w(T_{\ell})}{2n}
        =: \Delta_w(t) + \varepsilon.
    \end{align*}
    Then, it follows from the conditional Hoeffding bound (full version of \cite{DBLP:conf/podc/AmirABBHKL23}) for $\lambda = 2\shBias \cdot \Delta_w(T_{\ell})/2$ 
    \begin{align*}
        &\Probs{\Delta_w(T_{\ell}+\tau) < 2\Delta_w(T_{\ell})} \\
        &\phantom{={}}= \Probs{\Delta_w(T_{\ell}+\tau) < \Delta_w(T_{\ell}) + \tau \cdot \varepsilon - \lambda} \\
        &\phantom{={}}= \Probs{\sum_{t =T_{\ell}+1}^{T_{\ell}+\tau} \Delta_w(t) < \sum_{t =T_{\ell}}^{T_{\ell}+\tau-1} (\Delta_w(t) + \varepsilon) - \lambda}
        =\Probs{Z-\mu < \lambda}
        \leq e^{-\frac{2\lambda^2}{4(\shBias^2/6)n}}
        \leq n^{-3}.
    \end{align*}
    It follows by the union bound over the high probability events from above that
    \begin{align*}
        \Probs{\exists t \in [T_{\ell},T_{\ell}+\tau] \colon \Delta_w(t) \geq \min\set{2\cdot \Delta_w(T_{\ell}),(n-u(t))/4}} \geq 1-n^{-2}.
    \end{align*}
       
    Applied to a fixed $\ell$, this implies the length of the interval $\mathcal{I}_{\ell}$ is \whp at most $(\shBias^2/6)\cdot n$.
    From the union bound over all $\ell \leq \log n$ intervals we get there exists a time $t \in \intcc{T_0,T_{\log n}}$ such that $\Delta_w(t) \geq n-u(t)/4$.
    Otherwise, $\Delta_w(t) \geq 2^{\log n} \cdot \Delta_w(t_0) > n$, leading to a contradiction.
    The length of $\intcc{T_0,T_{(\log n)}}$ is at most $(\shBias^2/6) \cdot n\cdot \log n$.
    At last, it follows from \cref{lem:undecided_u_bound} that $(n-u(t))/4\geq n/10$.
\end{proof}

We are ready to prove \sref{thm:main-theorem:1} of \cref{thm:main-theorem}.
\begin{proof}[Proof of \sref{thm:main-theorem:1}]
    Consider a configuration with $x_1(0) \in \intcc{\epsilon \cdot n, x_2}, u(0) \leq n/2$ and $p-p_s = \Omega\left(n^{-1/2}\cdot \log n\right)$.
    Then equivalently $\Delta_w(0) = \Omega(x_2 \cdot n^{-1/2} \cdot \log n)$ and $x_1+x_2 \geq n/2$. Since $x_1 \leq x_2$, we have $x_2 = \Theta(n)$ and thus $\Delta_w(0) \geq \xi \cdot \sqrt{n \log n}$ for some constant $\xi$.
    
    Let $T_a = \inf\set{t\geq 0 ~|~ \Delta_w(t) \geq n/10 }$
    and $T_b = \inf\set{t\geq 0 ~|~ x_2(t) = 0 }$.
    By \cref{lem:repeated-doubling}, we have \whp $T_a = O(n\log n)$ and then by \cref{lem:drift_x2_to_zero} $T_b = T_a + O(n \log n)$.
    
    Note that at no time all agents can be undecided, since
    the last agent with Opinion 1 cannot encounter Opinion 2.
    Therefore, at time $T_b$, at least one agent with Opinion 1 exists.

    With $x_2(T_b) = 0$, the process simplifies to a single productive rule: $\delta(1,\bot) = 1$. 
    Let $T_1 = \inf\set{t \geq T_b ~|~ x_1(t) = n}$.
    Assuming that $T_b < \infty$ and $x_1(T_1) = 1$, we have $T_1 \leq T_b + 6 n \log n$ with probability at least $1-n^{-2}$ (see e.g., \cite{berenbrink2021loosely}).
    It is easy to see that $x_1(T_b)=1$ gives an upper bound for $T_1$.
    The statement then follows from the union bound.
\end{proof}

\section{Cases in which either Opinion 1 or Opinion 2 wins}
\label{sec:analysis:case3} In this section we consider the critical regime of $p \approx 1-x_1(0)/x_2(0)$, i.e., the initial weighted bias ($\Delta_w(0) \coloneqq x_1(0) - (1-p) \cdot x_2(0)$) is small.
We show that we reach a configuration after a while with a sufficiently large weighted bias ($|\Delta_w(t)| = \Omega(\sqrt{n\log n})$).
We do so by defining a submartingale $(Y_t)_{t\geq 0}$ as $Y_t = \Delta_w^2(t)-\submConst \cdot t$ for a suitably chosen constant $\submConst$ and applying tail bounds (\cref{lem:create-initial-bias}).
At this point, either \sref{thm:main-theorem:1} or \sref{thm:main-theorem:3} of \cref{thm:main-theorem} applies.

We define $T_w$ as the first time that the process reaches such a weighted bias. More formally, $T_w \coloneqq \inf\set{t \geq 0: \abs{\Delta_w(t)} \geq \shBias \cdot \sqrt{n \cdot \log n}}$.
We show that $T_w = O(n \cdot \log^2 n)$ (\cref{lem:create-initial-bias}).

One of the standard approaches is via a combination of anti-concentration bounds and concentration bounds for random walks.
It is rather complicated to define a random walk on the weighted bias $\Delta_w(t)$ since this results in a non-integer state space. 
Instead, we define two submartingales $(Z_t)_{t \geq 0}$ with $Z_t=\Delta_w(t)$ and $Y_t = Z_t^2 - \submConst\cdot t$. We show that $Y_t$ is a submartingale for a suitably chosen constant $\submConst$ and 
then we prove that $\E{Z_T-Z_0} = \sqrt{T}$ by 
considering $Y_t $.
To bound $|Y_{t+1}-Y_{t}|$ we will use tail bounds (see \cref{lem:azuma-hoeffding}).
Unfortunately, there is one more challenge we have to address.
We can only show that $Y_t$ is a submartingale as long as the number of undecided agents is not too large (which is shown in \cref{sec:undecided}).

\begin{lemma}
\label{lem:create-initial-bias}
    Let $\textbf{x}(0)$ be a configuration with $\abs{\Delta_w(0)} < \shBias \cdot \sqrt{n\cdot \log n}$ for an arbitrary constant $\xi$.
    Let $T_w \coloneqq \inf\set{t \geq 0: \abs{\Delta_w(t)} \geq \shBias \cdot \sqrt{n \cdot \log n}}$.
    Then w.h.p.\ $T_w = O(n \cdot \log^2 n)$.
\end{lemma}
\begin{proof}
    The idea of the proof is to apply the Azuma-Hoeffding bound to a suitable submartingale.
    Let $T_u \coloneqq \inf\set{t \geq 0: u(t) > x_1(t) + x_2(t) + 6 \shBias \cdot \sqrt{n \cdot \log n}}$.
    We define $Y_t = \Delta_w(t)^2 - \submConst \cdot t$ for $t < \min\set{T_w,T_u}$ where the constant $\submConst$ is chosen later.
    Otherwise, $Y_t = Y_{t-1}$.
    Note that by \cref{lem:undecided_u_bound}, \whp $T_u = \omega(n\log^2n)$.
    We show that $Y_0,Y_1,\ldots$ is a submartingale.
    The calculation is similar to that in the proof of \cref{lem:shifted_bias_bound:case1}.
    We consider every possible interaction and the resulting change.
    Assume that $t < \min\set{T_w,T_u}$. 
    Then,
    
    \begin{align*}
    \MoveEqLeft\E{Y_{t+1} | \filter_t, \Delta_w(t) = \Delta_w}\\
        &= \E{\Delta^2_s(t+1)-\submConst \cdot (t+1) | \filter_t, \Delta_w(t) = \Delta_w}\\
            &= \frac{x_1 \cdot x_2}{n^2} \cdot \left( p \cdot \Delta_w^2 + (1-p) \cdot (\Delta_w-1)^2 \right) + \frac{x_2 \cdot x_1}{n^2} \cdot (\Delta_w+1-p)^2 \\
            &\phantom{={}} + \frac{u \cdot x_1}{n^2} \cdot (\Delta_w+1)^2 + \frac{u \cdot x_2}{n^2} \cdot (\Delta_w-1+p)^2 \\
            &\phantom{={}} + \frac{n^2-2x_1\cdot x_2 - u\cdot (x_1+x_2)}{n^2} \cdot \Delta_w^2 - \submConst \cdot (t+1)\\
        &= \frac{x_1 \cdot x_2}{n^2} \cdot \left( 2\Delta_w^2 + (2-p)(1-p) \right)
        + \frac{u \cdot x_1}{n^2} \cdot \left( \Delta_w^2 + 2\Delta_w+1 \right) \\
        &\phantom{={}}+ \frac{u \cdot x_2}{n^2} \cdot \left(\Delta_w^2 -2(1-p)\Delta_w +(1-p)^2\right) \\
        &\phantom{={}} + \frac{n^2-2x_1\cdot x_2 - u\cdot (x_1+x_2)}{n^2} \cdot \Delta_w^2 - \submConst \cdot (t+1)\\
        &= \Delta_w^2 + \frac{x_1 \cdot x_2}{n^2} \cdot (2-p)\cdot (1-p)\\
            &\phantom{={}}+ \frac{u \cdot 2\Delta_w}{n^2} \cdot (x_1 - (1-p)x_2) + \frac{u}{n^2} \cdot \left( x_1 + (1-p)^2x_2\right) -\submConst\cdot (t+1) \\
        &= Y_t + \frac{x_1 \cdot x_2}{n^2} \cdot (2-p)\cdot (1-p)
        + \frac{u \cdot 2\cdot \Delta_w^2}{n^2} + \frac{u}{n^2} \cdot (x_1 + (1-p)^2 x_2) -\submConst\\
        &\geq Y_t + \frac{x_1 \cdot (1-p)x_2}{n^2} -\submConst.
    \end{align*}
    Note that $p$ is constant and for $t < \min\set{T_w,T_u}$ it holds that $u < x_1 + x_2 + o(n)$ and $x_1 - (1-p)x_2 = o(n)$.
    Then $x_1,(1-p)x_2 = \Theta(n)$.
    So there exists a worst case bound on $(x_1 \cdot (1-p)x_2)/n^2$ that is in $\Theta(1)$.
    Using that bound for $\submConst$ yields the desired result
    $\E{Y_{t+1} | \filter_t} \geq Y_t$.
    For $t \geq \min\set{T_w,T_u}$, we have $\E{Y_{t+1} | \filter_t} = Y_t$ by definition.

    Next, we apply the Azuma-Hoeffding bound (\cref{lem:azuma-hoeffding}) with 
    $\tau = \alpha \cdot n \cdot \log^2 n$ for some constant $\alpha$ that we determine later
    and $\lambda = \Delta_w^2(0) + \shBias \cdot \sqrt{10 \cdot \tau \cdot n} \cdot \log n \geq \shBias \cdot \sqrt{10 \cdot \tau \cdot n} \cdot \log n$.
    Furthermore, for $t < \min\set{T_w,T_u}$ we have
    \begin{align*}
        \MoveEqLeft |Y_{t+1}-Y_t| 
        \leq (\Delta_w(t)+1)^2 - \submConst \cdot (t+1)-\Delta_w^2(t) + \submConst \cdot t
        = 2 \Delta_w(t) + 1 - \submConst\\
        &\leq 2 \cdot \shBias \cdot \sqrt{n \cdot \log n} + 1.
    \end{align*}
    Thus we have $(b-a)^2 \leq \left(2 \cdot (2\cdot \shBias \cdot \sqrt{n \cdot \log n} + 1)\right)^2 \leq 5 \cdot \shBias^2 \cdot  n \cdot \log n$.
    This yields probability 
    \begin{align*}
        \MoveEqLeft\exp\left(- \frac{2\lambda^2}{\tau (b-a)^2} \right)
        \leq \exp\left(- \frac{2 \cdot \shBias^2 \cdot 10 \cdot \tau \cdot n \cdot \log^2 n}{\tau \cdot 5 \cdot \shBias^2 \cdot n \cdot \log n}\right)
        = \exp\left(- 4\cdot \log n\right)
        = n^{- 4}
    \end{align*}
    for the event
    \begin{align*}
    &\Delta_w^2(\tau) - \submConst\cdot \tau - (\Delta_w^2(0) - 0) < - \lambda\\
    \Leftrightarrow &\abs{\Delta_w(\tau)} < \sqrt{\submConst\cdot \tau + \Delta_w^2(0) - \lambda}\\
    \Leftrightarrow &\abs{\Delta_w(\tau)} < \sqrt{\submConst\cdot \tau - \shBias \cdot \sqrt{10\cdot \tau \cdot n} \cdot \log n}\\
    \Leftrightarrow &\abs{\Delta_w(\tau)} < \sqrt{\submConst\cdot \alpha \cdot n \cdot \log^2 n - \shBias \cdot \sqrt{10\cdot \alpha } \cdot n \cdot \log^2 n}\\
    \Leftrightarrow &\abs{\Delta_w(\tau)} < \sqrt{\submConst\cdot \alpha - \shBias \cdot \sqrt{10\cdot \alpha}} \cdot \sqrt{n} \cdot \log n.
    \end{align*}
    It is now clear that there exists a constant $\alpha = \alpha(\shBias,\submConst) > 0$ such that $\submConst\cdot \alpha - \shBias \cdot \sqrt{10 \cdot \alpha} \geq \shBias^2$.
    Then, $\abs{\Delta_w(\tau)} \geq \sqrt{\submConst\cdot \alpha - \sqrt{\submConst\cdot \alpha}} \cdot \sqrt{n} \cdot \log n$ implies $\abs{\Delta_w(\tau)} \geq \shBias \cdot \sqrt{n \cdot \log n}$, i.e., $T_w \leq \tau$.
\end{proof}

\section{Proof of \cref{thm:coupling-theorem}}
\label{sec:coupling}
In this section we define a coupling to show \cref{thm:coupling-theorem}. 
To do so we first define an order on the state space. A configuration $\xconfig$ is \emph{better} 
than a configuration $\tconfig$ if it does not have fewer agents in state $1$ and not more agents in state $2$. 
Then we prove \cref{thm:coupling-theorem} 
which states that increasing the stubbornness or starting from a better configuration decreases the time until all agents agree on Opinion 1. 
This allows us to extend high probability results to more configurations and processes (see \cref{fig:proof-overview}).

We assume that the state space $Q$ is ordered. The order $\better$ is defined as $1 \better \bot \better 2$.
For configurations $\xconfig=(x_1,x_2, x_3), \tconfig=(\tilde{x}_1,\tilde{x}_2,\tilde{u})$ we say $\xconfig \dominates \tconfig$ ($\xconfig$ is \emph{better} than $\tconfig$) if $x_1 \geq \tilde{x}_1$ and $x_1+u \geq \tilde{x}_1+\tilde{u}$.
We first observe a property of $\USD{p}$ in  \cref{obs:monotonous} namely that one interaction of $\USD{p}$ with agents in better states yield a ``better'' outcome.

\begin{observation}
    \label{obs:monotonous}
    For all $q_i,q_j,q_{i'},q_{j'} \in Q$ with $q_i \better q_{i'}$ and  $q_j \better q_{j'}$ we have 
    \begin{align*}
        \usd{p}(q_i,q_j) &\better \usd{p}(q_{i'},q_{j'}) &\text{ if } (q_i,q_j) \neq (1,2) \text{ and } (q_{i'},q_{j'}) \neq (1,2)\\
        \Pr[\usd{p}(q_i,q_j)=1] &\geq \Pr[\usd{p}(q_{i'},q_{j'})=1] &\text{ otherwise}
    \end{align*}
    %We say $\USD{p}$ is \emph{\monotonous}.
\end{observation}

\begin{proof}
    The claim can be verified by checking all cases. 
    We visualize all possible interactions in \cref{fig:coupling-cases}.
    $\better$ induces a partial order on these interactions.
    It can be verified in \cref{fig:coupling-cases}, that along this partial order, $\better$ is preserved on the outcome of the interactions. 
    All remaining cases follow from transitivity of $\better$.
\end{proof}
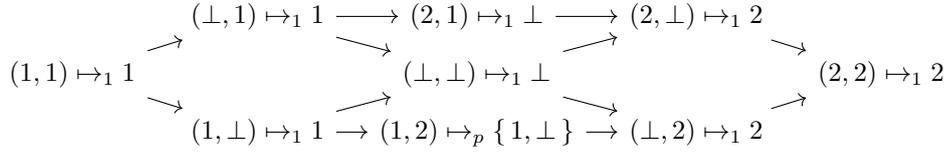
\begin{figure}[H]
    \centering
    \[\begin{tikzcd}[sep=tiny]
        && {(\bot,1)\mapsto_11} && {(2,1)\mapsto_1\bot} && {(2,\bot)\mapsto_12} \\
        {(1,1)\mapsto_11} &&&& {(\bot,\bot)\mapsto_1\bot} &&&& {(2,2)\mapsto_12} \\
        && {(1,\bot)\mapsto_11} && {(1,2)\mapsto_p\set{1,\bot}} && {(\bot,2)\mapsto_12}
        \arrow[from=2-1, to=1-3]
        \arrow[from=2-1, to=3-3]
        \arrow[from=1-3, to=1-5]
        \arrow[from=1-3, to=2-5]
        \arrow[from=3-3, to=2-5]
        \arrow[from=3-3, to=3-5]
        \arrow[from=1-5, to=1-7]
        \arrow[from=2-5, to=1-7]
        \arrow[from=2-5, to=3-7]
        \arrow[from=3-5, to=3-7]
        \arrow[from=3-7, to=2-9]
        \arrow[from=1-7, to=2-9]
    \end{tikzcd}\]%
    \caption{An overview over all possible interactions. All but one interaction have a deterministic outcome which we indicated by $\mapsto_1$.
    The interaction $(1,2)$ produces Opinion $1$ with probability $p$ and $\bot$ with $1-p$.
    We indicated this by $\mapsto_p \set{1,\bot}$.}
    \label{fig:coupling-cases}
\end{figure}%

In the next lemma we apply \cref{obs:monotonous} to show the preservation of better configurations using a coupling of  configurations $\xconfig, \tconfig$ with $\xconfig \dominates \tconfig$. 
To do so we assume that, additionally to a random pair (initiator, responder), a number $r \in \intoc{0,1}$ is chosen uniformly at random. If the initiator has opinion $1$ it determines the outcome of the interaction: $\usd{p}(1,2)$ equals 1 if $r \leq p$ and $\bot$ otherwise. In all other cases the transitions are performed with probability one (for all choices of $r$). 

For the sake of the argument we assume in the following that the agents are identified with labels $1$ through $n$.
Then the random choice in each interaction is an element of $[n]^2\times\intoc{0,1}$.
We now consider a fixed time $t$ and denote the state of agent $i$ in configuration $\xconfig(t)$ by $q_i(t)$ (or $\tilde{q}_i(t)$ for $\tconfig(t)$). W.l.o.g.\ we assume that agents are sorted s.t.\ $\forall i \in [n-1]: q_i(t) \better q_{i+1}(t)$ and $\tilde{q}_i(t) \better \tilde{q}_{i+1}(t)$. 

\begin{lemma}
    \label{lem:coupling-lemma}   
    Consider $\USD{p}(\xconfig)$ and $\USD{\tilde{p}}(\tconfig)$ with $p \geq \tilde{p}$ and $\xconfig(t) \dominates \tconfig(t)$.
    Then there exists a bijection $b: [n]^2\times\intoc{0,1}\rightarrow[n]^2\times\intoc{0,1}$ for the random choices of the schedulers such that $\xconfig(t+1) \dominates \tconfig(t+1)$.
\end{lemma}

\begin{proof}
    Assume that the random scheduler of $\USD{p}$ draws initiator $i$, responder $j$ and random number $r$.
    Using identical random choices for the initiator and the responder for $\USD{\tilde{p}}$,
    we have $q_i(t) \better \tilde{q}_i(t)$ and $q_j(t) \better \tilde{q}_j(t)$.
    Then by \cref{obs:monotonous}, for $(q_i,q_j) \neq (1,2)$ and $(q_{i'},q_{j'}) \neq (1,2)$ we have 
    $\usd{p}(q_i,q_j) \better \usd{p}(q_{i'},q_{j'})$ and otherwise 
    \begin{align*}
        \Pr[\usd{p}(q_i(t),q_j(t))=1]
        \geq \Pr[\usd{p}(\tilde{q}_{i}(t),\tilde{q}_{j}(t))=1] 
        \geq \Pr[\usd{\tilde{p}}(\tilde{q}_{i}(t) = 1].
    \end{align*}
    Then, by using the same random choice for the $r$, we have $\usd{p}(q_i(t),q_j(t)) \better \usd{\tilde{p}}(\tilde{q}_{i}(t),\tilde{q}_{j}(t))$.
    Then we have $q_i(t+1) \better \tilde{q}_i(t+1)$ for the initiator. No other agent is changing its state. 
    This implies $\xconfig(t+1) \dominates \tconfig(t+1)$.    
\end{proof}

\Cref{thm:coupling-theorem} (restated here for convenience) follows by repeatedly applying \cref{lem:coupling-lemma}.
\couplingtheorem*

\begin{proof}
    We prove the more general claim that there exists a coupling $([n]^2 \times \intoc{0,1})^t \rightarrow ([n]^2 \times \intoc{0,1})^t$ (mapping the random choices of the scheduler of $\USD{p}$ to those of the random scheduler of $\USD{\tilde{p}}$) such that for all $t \geq 0: \Xconfig(t) \dominates \Tconfig(t)$.
    The claim holds for $t=0$ since $\xconfig \dominates \tconfig$.
    The induction step from $t$ to $t+1$ follows from \cref{lem:coupling-lemma}.

    The two statements of the theorem follow from the fact that $\tilde{X}_1(t) = n \land \Xconfig(t) \dominates \Tconfig(t) \implies X_1(t) = n$ and that $X_2(t) = n \land \Xconfig(t) \dominates \Tconfig(t) \implies \tilde{X}_2(t) = n$.
\end{proof}
\section{Conclusion}
\label{sec:conclusion}
We have extended the USD for two opinions by introducing the stubbornness parameter $p$ that governs how likely agents of Opinion 1 lose their opinion.
We have shown that the outcome of the process depends on the \emph{weighted bias} $\Delta_w(t) \coloneqq x_1(t) - (1-p)x_2(t)$.
More specifically, there exists a phase transition around $\Delta_w(t) = 0$ where the remaining opinion switches.
Under mild constraints on $p$, $x_1$ and $u$ we have shown that the agents reach a consensus in $O(n\log n)$ interactions for $\abs{\Delta_w(0)} = \Omega(\sqrt{n\log n})$ which generalizes the known result for $p=0$.
Moreover, we have shown for $\abs{\Delta_w(0)} = o(\sqrt{n\log n})$ that the agents reach consensus on some opinion in $O(n\log^2 n)$ interactions.
It is noteworthy that this threshold is not directly depending on $u(t)$.

\paragraph{Open questions}
It remains an open question if the higher runtime of $O(n\log^2 n)$ is an artifact of the analysis or if it is due to the additional randomness making the process less predictable.
Our analysis covers the case of exactly two opinions where one of the opinions is retained with a certain probability.
For the analysis of more opinions, the first obstacle is to find the proper generalization of how stubborn agents are in which interactions.
The second task is then to find a generalized description of the phase transition which is already not trivial without stubborn agents.
The analysis of \cref{sec:coupling} can be extended to many similar protocols with a preferred opinion -- altering the transition $(\bot,2) \rightarrow 2$ of the USD to only adopt Opinion 2 with probability $p$ is one example.
However, this process has a more complicated threshold function that depends on the number of undecided agents.
Therefore, it is an interesting open problem to find the  corresponding threshold for initial configurations without undecided agents.

\newpage
% \printbibliography
\bibliography{references}

\newpage
\appendix

\section{Auxiliary Results}
\label{apx:auxiliary-results}
This appendix states several auxiliary results we use throughout our analysis.
\begin{theorem}[\cite{DBLP:series/ncs/Lengler20}]
    \label{thm:mult_drift_tail_lengler_18}
    Let $(X_t)_{t\geq0}$ be a sequence of non-negative random variables with a finite state space $\mathcal{S} \subseteq \mathbb{R}_0^+$ such that $0 \in \mathcal{S}$. Let $s_{min} \coloneqq \min(\mathcal{S}\setminus\{0\})$, and let $T \coloneqq \inf\{t\geq 0 ~|~ X_t = 0\}$. Suppose that $X_0 = s_0$, and that there exists $\delta > 0$ such that for all $s \in \mathcal{S}\setminus\{0\}$ and all $t\geq 0$, 
    \begin{align*}
        E[X_t-X_{t+1}~|~X_t = s] \geq \delta \cdot s.
    \end{align*}
    Then, for all $r \geq 0$,
    \begin{align*}
        Pr\left[ T > \left \lceil{\frac{r+\ln(s_0/s_{min})}{\delta}} \right \rceil \right] \leq e^{-r}.
    \end{align*}
\end{theorem}

\begin{lemma}[\cite{feller68}]
\label{lem:random_walk_no_ruin_probability}
If we run an arbitrarily long sequence of independent trials, each with success probability at least $p$, then the probability that the number of failures ever exceeds the number of successes by $b$ is at most $((1-p)/p)^b$.
\end{lemma}

\begin{lemma}[\cite{DBLP:journals/corr/abs-2402-06471}]
\label{lem:random-walk}
Let $(W_t)_{t \in \mathbb{N}}$ be a biased random walk on state space $\mathbb{N}_0$, initially at $0$. Let $0 < p < 1$ denote the probability for the walk to move to the right (increase its current position by 1). Conversely, let $q=(1-p)$ denote the probability that it moves to the left (or stays in position in case it currently resides at position $0$). Then, for any $N>0$ and hitting time $\tau_N = \min\{t ~|~ W_t = N\}$ the following holds:
\begin{enumerate}
    \item If $p>q$ then $\tau_N \leq (\frac{2}{p-q})^2 \cdot N$ with probability at least $1-\exp(-N)$.
    \item If $p<q$ then $\tau_N \geq (q/p)^{N/2}$ with probability at least $1-(p/q)^{N/2}$. 
\end{enumerate}
\end{lemma}

\begin{lemma}[\cite{DBLP:conf/podc/AmirABBHKL23}]
    \label{lem:biased_reflective_random_walk}
    Let $W(t)$ be the random variable at time $t$ of a random walk on the positive integers with a reflective border at $0$ and $W(0) = 0$.
    Let $p$ be the probability of a $+1$-step.
    Let $q>p$ be the probability of a $-1$-step everywhere except for the origin.
    Let $r=1-p-q$ be the probability of remaining in place ($1-p$ for the origin).
    Let $T_m = \inf\set{t \geq 0 ~|~ W(t) \geq m}$.
    Then $\Pr[T_m \leq n^c] \leq n^c \cdot (p/q)^m$.
\end{lemma}

\begin{lemma}[\cite{DBLP:books/daglib/0015598}]
    \label{lem:azuma-hoeffding}
    Consider a submartingale $Z_0,Z_1,\ldots$ w.r.t.\ a filtration $\mathcal{F} = (\mathcal{F}_i)_{i=0}^{\tau-1}$.
    Assume $a\leq Z_i - Z_{i-1} \leq b$ for all $i \geq 1$.
    Then for all positive integers $\tau$ and $\lambda > 0$
    \begin{align*}
        \Pr[Z_\tau - Z_0 <  - \lambda] \leq e^{-\frac{2\lambda^2}{\tau(b-a)^2}} 
    \end{align*}
\end{lemma}

\begin{theorem}[\cite{Janson17}]
\label{lem:janson}
Let $X = \sum_{i=1}^{n} X_i$ where $X_i, i = 1,\dots,n$, are independent geometric random variables with $X_i \sim Geo(p_i)$ for $p_i \in (0,1]$.
For any $\lambda \geq 1$,
\begin{equation*}
    \Prob{X \geq \lambda \cdot \Ex{X}} \leq \exp({- \min_{i}\{p_i\} \cdot \Ex{X} \cdot (\lambda - 1 - \ln{\lambda})}).
\end{equation*}
\end{theorem}

\section{Further Results}
\label{apx:further-results}
\subsection{Bounding the number of undecided agents}
\label{sec:undecided}
In this section we show two results. First we show that, from an an arbitrary initial state we quickly reach a configuration $\mathbf{X}(t)$ with
$$\min\set{x_1(t),(1-p)x_2(t)}\le u(t)\le x_1(t)+x_2(t).$$
We show that, after entering that region, the number of undecided agents will w.h.p. remain in that region for $\Omega(n\log^2 n)$ interactions.

\begin{lemma}
    \label{lem:undecided_u_bound}
    Let $\UpperUBound(t) \coloneqq u(t)-x_1(t)-x_2(t)$ and $T_u \coloneqq \inf\set{t \geq 0: \UpperUBound(t) > 6 \shBias \cdot \sqrt{n\log n}}$. Let $\textbf{x}(0)$ be an arbitrary configuration with $\UpperUBound(0) < 2 \shBias \cdot \sqrt{n \log n}$.
    Then w.h.p.\ it holds that $T_u = \omega(n \log^2 n)$.
\end{lemma}

\begin{proof}
    The idea is to show for any time $t > 0$ where $u(t)$ crosses the threshold $x_1(t) + x_2(t) + \Delta$ (for $\Delta \coloneqq 2\shBias \cdot \sqrt{n \cdot \log n}$) from below, it only ever exceeds it by another $2\Delta$ due to a negative drift that we now calculate.
    We say the $t$'th interaction is \emph{productive}, if $\Xconfig(t+1) \neq \Xconfig(t)$.
    We denote this event by $\productive$.
    Note that any productive interaction alters $\UpperUBound$ by $\pm 2$.
    Then,
    \begin{align*}
        \Probs{\UpperUBound(t+1) = \UpperUBound(t)+2 ~|~ \filter_t}
        &= \frac{(1-p)x_1\cdot x_2}{n^2} + \frac{x_2\cdot x_1}{n^2}
        &&= \frac{(2-p)x_1\cdot x_2}{n^2},\\
    %    &= \frac{1}{2} + \frac{(2-p)x_1\cdot x_2-u\cdot (x_1+x_2)}{2(u\cdot (x_1+x_2) +(2-p)x_1\cdot x_2)} \\
        \Probs{\UpperUBound(t+1) = \UpperUBound(t)-2 ~|~ \filter_t}
        &= \frac{u\cdot x_1}{n^2} + \frac{u\cdot x_2}{n^2}
        &&= \frac{u\cdot (x_1+x_2)}{n^2}.
    %    &= \frac{1}{2} + \frac{u\cdot (x_1+x_2)-(2-p)x_1\cdot x_2}{2(u\cdot (x_1+x_2) +(2-p)x_1\cdot x_2)} \\
    \end{align*}
    
    Note that for a fixed $u$, the product $x_1 \cdot x_2$ is maximal for $x_1 = x_2 = (x_1+x_2)/2$.
    Then,
    \begin{align*}
        \MoveEqLeft \Probs{\UpperUBound(t+1) = \UpperUBound(t)+2 ~|~ \mathcal{F}_t, \productive}
        = \frac{(2-p)x_1\cdot x_2}{(2-p)x_1\cdot x_2 + u \cdot (x_1+ x_2)}\\
        &= \frac{1}{2} + \frac{(2-p)x_1\cdot x_2 - u \cdot (x_1+ x_2)}{2\left((2-p)x_1\cdot x_2 + u \cdot (x_1+ x_2)\right)}
        \leq \frac{1}{2} + \frac{(2-p)\frac{x_1+x_2}{2}\cdot \frac{x_1+x_2}{2} - u \cdot (x_1+ x_2)}{2\left((2-p)x_1\cdot x_2 + u \cdot (x_1+ x_2)\right)}\\
        &= \frac{1}{2} + \frac{(x_1+ x_2)\cdot \left(\frac{2-p}{4}\cdot (x_1+x_2)-u\right)}{2\left((2-p)x_1\cdot x_2 + u \cdot (x_1+ x_2)\right)}
        \leq \frac{1}{2} - \frac{(x_1+ x_2)\cdot \UpperUBound(t)}{2\left((2-p)x_1\cdot x_2 + u \cdot (x_1+ x_2)\right)}
    \end{align*}
    
    So as long as $\UpperUBound(t) \geq \Delta > 0$ holds at time $t$, we can bound the denominator by
    \begin{align*}
        \MoveEqLeft 2\left((2-p)x_1\cdot x_2 + u \cdot (x_1+ x_2)\right)\\
        &\leq 2\left(2x_1\cdot x_2 + u \cdot (x_1+ x_2)\right)\\
        &\leq 2((x_1+x_2)^2+u \cdot (x_1+x_2))\\
        &= 2(u+x_1+x_2) \cdot (x_1+x_2)\\
        &= 2n \cdot (x_1+x_2).
    \end{align*}
    Therefore, if $\UpperUBound(t) \geq \Delta$ and we condition on a productive step, we have
    \begin{align*}
        \Probs{\UpperUBound(t+1) = \UpperUBound + 2 ~|~ \mathcal{F}_t, \productive}
        \geq \frac{1}{2} - \frac{\Delta \cdot (x_1+x_2)}{2n \cdot (x_1+x_2)}
        \geq \frac{1}{2} - \frac{\Delta}{2n}.
    \end{align*}

    Now, we consider a sequence of $m= \omega(n \cdot \log^2 n)$ interactions.
    We let $T_i$ for $i>0$ denote the first time after $T_{i-1}$ where $\UpperUBound(T_i-1) < \Delta$ and $\UpperUBound(T_i) \geq \Delta$.
            We show that in every sequence of at most $m$ interactions, starting at time $T_i$, $\UpperUBound$ does not exceed the threshold by more than $3\Delta$ before falling back below $\Delta$ with probability at least  $1-n^{-4\shBias^2}$.
    To show that, we consider a sequence of independent Bernoulli trials $(Z_i)_{i>0}$ with success probability $\tilde{p} = 1/2 + \Delta / (2n)$.
    Each trial corresponds to a productive interaction where $\UpperUBound(t) \geq \Delta$ and a success to a $-2$-step in $\UpperUBound$. 
    The number of failed trials exceeds the number of successful trials by more than $\Delta$ trials -- and thus $\UpperUBound$ exceeds $\Delta + 2\Delta = 6 \shBias \cdot \sqrt{n\cdot \log n}$ with probability at most (see \cref{lem:random_walk_no_ruin_probability}):
    \begin{align*}
        \left(\frac{1-\Tilde{p}}{\Tilde{p}}\right)^{\Delta}
        %= \left( \frac{n-\Delta}{n+\Delta} \right)^{\Delta}
        = \left(1-\frac{2\Delta}{n+\Delta}\right)^{\Delta}
        \leq \exp\left(-\frac{2\Delta^2}{n+\Delta}\right)
        \leq \exp\left(-\frac{2\cdot 4\shBias^2 \cdot n \cdot \log n}{2n}\right)
        = n^{-4\shBias^2}.
    \end{align*}
    At last, every such sequence has at least one interaction, thus there are at most $m$ sequences in $m$ interactions.
    The claim then follows from the union bound.
\end{proof}

We complete this section by showing that after creating at least $\min\set{x_1,(1-p)x_2}$ undecided agents, that number does not drop significantly while the weighted bias is not large enough.
The proof idea is the same as in \cref{lem:undecided_u_bound}.
\begin{lemma}
\label{lem:lower_bound_u}
    Let $T = \inf\set{t\geq 0 ~|~ u(t) \geq \min\set{x_1(t),(1-p)x_2(t)}}.$
    Let $\xconfig(0)$ be an arbitrary initial configuration with $u(0) \leq x_1(0)+x_2(0)$. 
    Then, $\Probs{T \leq  144\cdot n}\geq 1-n^{-3}$.
    Furthermore, as long as $\sqrt{n\log n} \leq \Delta_w(t) \leq (n-u(t))/4$ for $t\geq T$, it holds that $u(t) \geq (3/10) \cdot (1-p)/(2-p) \cdot n$ \whp.
\end{lemma}
 \begin{proof}
    To show the result we track the evolution of the undecided agents over time.
    Let $T = \inf\set{t\geq 0 ~|~ u(t) \geq \min\set{x_1(t),(1-p)x_2(t)}/2}$.
    We distinguish between $\min\set{x_1(t),(1-p)x_2(t)} = x_1(t)$ and $\min\set{x_1(t),(1-p)x_2(t)} = (1-p)x_2(t)$.
    In the first case, as long as $t < T$, it holds that
    \begin{align*}
        \Probs{U(t+1)=u(t)+1 ~|~ \filter_t, \productive} 
        &= \frac{1}{2} + \frac{(2-p)x_1 x_2 - u\cdot(x_1+x_2)}{2((2-p)x_1 x_2+u\cdot(x_1+x_2))} \\
        &\geq \frac{1}{2} + \frac{(2-p)(x_1+x_2) - x_1/2 \cdot (x_1+x_2)}{2((2-p)x_1 x_2+u\cdot(x_1+x_2))} \\
        &\geq \frac{1}{2} + \frac{x_1((x_1+x_2)/2-\Delta_w)}{3x_1(x_1+x_2)} \\
        &\geq \frac{1}{2} + \frac{1}{12} .
    \end{align*}
    In the second case, as long as $t<T$, it holds that
    \begin{align*}
        \Probs{U(t+1)=u(t)+1 ~|~ \filter_t, \productive} 
        &= \frac{1}{2} + \frac{(2-p)x_1 x_2 - u\cdot(x_1+x_2)}{2((2-p)x_1 x_2+u\cdot(x_1+x_2))} \\
        &\geq \frac{1}{2} + \frac{(2-p)(x_1+x_2) - (1-p)/2 x_2 \cdot (x_1+x_2)}{2((2-p)x_1 x_2+u\cdot(x_1+x_2))} \\
        &\geq \frac{1}{2} + \frac{x_2((1-p)(x_1+x_2)/2-\Delta_w)}{3x_2(x_1+x_2)} \\
        % &\geq \frac{1}{2} + \frac{(1-p)(x_1+x_2)/4)}{3(x_1+x_2)} \\
        &\geq \frac{1}{2} + \frac{1}{12} .
    \end{align*} 
    
    Now we bound the number of productive interactions in the first $O(n)$ interactions.
    An interaction at time $t$ is productive with at least constant probability
    \begin{align*}
        \frac{(2-p)x_1(t)x_2(t)+u(t)(n-u(t))}{n^2} \geq \varepsilon
    \end{align*}
    due to \cref{lem:undecided_u_bound} and the assumption $\Delta_w(t) \leq (n-u(t))/4$.
    From Chernoff bound it follows  $\Omega(n)$ productive interactions in $O(n)$ interactions \whp.
    Now we show that the number of undecided agents does not drop below $\min\set{x_1(t)/2,(1-p)x_2/2}-\Delta$ with $\Delta = \sqrt{n\log n}$.
    This allows us to relate this sequence of productive interactions and the evolution of the undecided agents to a biased random walk on $\mathbb{N}_0$ with a reflective barrier at position $0$.
    The probability to move to the right  is $\tilde{p} = 1/2+1/12$
    and the probability to move to the left is $1-\tilde{p}$.
    Then \cref{lem:biased_reflective_random_walk} implies $T=O(n)$ \whp.
    
    \medskip
    Now we show that the number of undecided agents remains relatively large.
    Similar to \cref{lem:undecided_u_bound}, we consider the time steps $T_i$ where for the first time after $T_{i-1}$ we have $u(T_i-1) > \min\set{x_1(T_{i-1})/2,(1-p)x_2(T_{i-1})/2}$ and $u(T_i)\leq \min\set{x_1(T_i)/2,(1-p)x_2(T_i)/2 }$.
    We do a case study on $u(t)$.
    As long as $u(t) \leq x_1(t)/2$ for $t\geq T_i$ it holds that
     \begin{align*}
        \Probs{U(t+1)=u(t)-1 ~|~ \filter_t, \productive} 
        &= \frac{1}{2} - \frac{(2-p)x_1 x_2 - u\cdot(x_1+x_2)}{2((2-p)x_1 x_2+u\cdot(x_1+x_2))} \\
        &\leq \frac{1}{2} - \frac{(2-p)(x_1+x_2) - x_1/2 \cdot (x_1+x_2)}{2((2-p)x_1 x_2+u\cdot(x_1+x_2))} \\
        &\leq \frac{1}{2} - \frac{x_1((x_1+x_2)/2-\Delta_w)}{3x_1(x_1+x_2)} \\
        &\leq \frac{1}{2} - \frac{1+p}{12} \\
        &\leq \frac{1}{2} - \frac{1-p}{12}
    \end{align*}
    where in the last inequalities we use $\Delta_w <(1-p)(x_1+x_2)/4$.
    Similarly, as long as $u(t) \leq (1-p)x_2(t)/2$ for $t \geq T_i$ it holds that
    \begin{align*}
        \Probs{U(t+1)=u(t)-1 ~|~ \filter_t, \productive} 
        &= \frac{1}{2} - \frac{(2-p)x_1 x_2 - u\cdot(x_1+x_2)}{2((2-p)x_1 x_2+u\cdot(x_1+x_2))} \\
        &\leq \frac{1}{2} - \frac{(2-p)(x_1+x_2) - (1-p)/2 x_2 \cdot (x_1+x_2)}{2((2-p)x_1 x_2+u\cdot(x_1+x_2))} \\
        &\leq \frac{1}{2} - \frac{x_2((1-p)(x_1+x_2)/2-\Delta_w)}{3x_2(x_1+x_2)} \\
        &\leq \frac{1}{2} - \frac{(1-p)(x_1+x_2)/4)}{3(x_1+x_2)} \\
        &\leq \frac{1}{2} - \frac{1-p}{12}
    \end{align*} 
    where in the last inequalities we use $\Delta_w < (1-p)(x_1+x_2)/4$.
    Now we show that the number of undecided agents does not drop below $\min\set{x_1(t)/2,(1-p)x_2/2}-\Delta$ with $\Delta = \sqrt{n\log n}$.
    We consider a sequence of independent Bernoulli trials $(Z_i)_{i\geq 0}$ with success probability $\Tilde{p} = 1/2+ (1-p)/12$.
    Each trial corresponds to a productive interaction under the assumption from above.
    From \cref{lem:random_walk_no_ruin_probability} it follows the number of failed trials exceeds the number of successful trials by more than $\Delta/2$ trials with probability at most 
    \begin{align*}
        \left(\frac{1-\Tilde{p}}{\Tilde{p}}\right)^{\Delta/2}
        %= \left( \frac{n-\Delta}{n+\Delta} \right)^{\Delta}
        = \left( \frac{6-(1-p)}{6+(1-p)}\right)^{\Delta/2}
    \end{align*}
    At last, from $u\geq x_1/2-\Delta$, $x_1-(1-p)x_2 > 0$ and $n = x_1+x_2+u$ it follows
    $u \geq (3/10) \cdot (1-p)/(2-p) \cdot n$.
    On the other hand, from $u\geq (1-p)x_2/2-\Delta$, $x_1-(1-p)x_2>0$ and $n = x_1+x_2+u$ it follows $u\geq (3/10) \cdot (1-p)/(2-p)\cdot n$.
    The claim then follows from the union bound.
    \end{proof}

\subsection{Initially small Opinion 1}
\label{ssec:p=1}
Consider the configuration with $x_1(0) = o(n)$ and $x_2(0) = n-x_1(0)$.
For simplicity assume $x_1(0) = 1 $ and $x_2(0) = n-1$.
Observe that $x_1(0)/x_2(0) = o(1)$.
Therefore, we consider $p = 1$.
On an intuitive level, it is clear that the Opinion 1 wins eventually because it cannot lose support.
On the other hand, it is quite slow in the early stage due to the low number of undecided agents and agents with Opinion 1.
To increase the support of Opinion 1, an undecided agent has to interact with an agent with Opinion 1.
An undecided agent is created whenever an agent with Opinion $2$ interacts with an agent with Opinion $1$.
It adopts an opinion again, when interacting with an agent of either opinion. We call such an interaction \emph{$u$-productive}.
Considering the initial configuration, these interactions rarely occur in the beginning.

The proof idea for this case is straightforward.
First, we increase the support of Opinion 1 to a sufficiently large size, i.e., $x_1(t) = \Omega(\sqrt{n})$. 
Then we follow the proof of \sref{thm:main-theorem:1} of \cref{thm:main-theorem}.

\begin{lemma}
    Let $p = 1$ and $\mathbf{x}(0)$ be an initial configuration with $u(0)=0$, $x_1(0)=1$ and $x_2(0) = n-1$.
    Let $T_1 = \inf\set{t'\geq 0 ~|~ x_2(t') = 0}$.
    Then, $\Probs{T \leq 7n^2\log^2 n} \geq 1-n^{-2}$.
\end{lemma}
\begin{proof}
    We track the support of Opinion 1 until it reaches $2\sqrt{n}$ and then apply a drift theorem on the potential function $\Psi(t)= x_2(t)/x_1(t)$ similar to \cref{lem:drift_x2_to_zero}.
    First we show there exists a time $t = O(n^2\log^2 n)$ with $x_1(t) \geq 2\sqrt{n}$.
    Let $T = \inf\set{t\geq 0 ~|~ x_1(t)= 2\sqrt{n}}$.
    Assume $x_2(t) \geq n/2$ for all $t \leq T$ (otherwise, we are already done).
    Due to $p = 1$, $(2,1)$-interactions are the only interactions to create an undecided agent at time $t$ with probability $x_1(t)x_2(t) / n^2$.
    Let the random variable $Z$ denote the number of interactions until an undecided agent is created.
    Observe that $Z$ stochastically dominates a geometrically distributed random variable with success probability $1/(2n)$ because $x_2(t)\geq n/2$ and $x_1(t) \geq x_1(0)$.
    Therefore, $\Probs{Z \leq 8 n\log n} \geq 1-n^{-3}$.
    Similarly, we can describe the random variable $Y$ regarding $u$-productive interactions.
    The random variable $Y$ denotes the number of interactions until a $u$-productive interaction occurs, assuming $u(t)\geq 1$.
    The probability of such interaction is $u(t)(x_1(t)+x_2(t))/n^2$.
    Again, $Y$ stochastically dominates a geometrically distributed random variable with success probability $1/(2n)$.
    Therefore, $\Probs{Y \leq 8 n\log n} \geq 1-n^{-3}$.
    At last, assume a $u$-productive interaction occurs.
    In particular, the probability of an $(u,1)$-interaction is $x_1(t)/(x_1(t)+x_2(t))$, i.e., the support of the preferred opinion increases.
    Let the random variable $V$ denote the number of productive $(u,1)$-interactions until a $(u,1)$-interaction occurs.
    Again, $V$ stochastically dominates a geometrically distributed random variable with success probability $1/n$.
    Therefore, $\Probs{V \leq 8 n\log n} \geq 1-n^{-3}$.
    By combining these steps via the Union bound, the support of the preferred opinion increases by at least one within $O(n^2\log^2 n)$ interactions whp.
    At last, we repeat this at most $\sqrt{n}$ many times, and the first part follows by the Union bound.

    \medskip

    Now we deal with the second part.
    Similar to the proof of \cref{lem:drift_x2_to_zero}, we compute the expected change of the potential function $\Psi(t) = x_2(t)/x_1(t)$ and apply a known drift theorem.
    Due to the choice of stubbornness $p=1$ and the assumption we know $x_1(t')\geq x_1(t) = 2\sqrt{n}$ for all $t'\geq t$.
    This allows us to bound the expected change as follows
    \begin{align*}
        \Ex{\Psi(t+1)-\Psi(t)~|~ \filter_t}
         &= -\frac{\Psi(t)}{n^2 x_1 (x_1 +1)} \cdot (x_1^2+x_1 -u) \\
         &\leq - \frac{\Psi(t)}{n^2} \cdot (x_1 - \frac{n}{x_1+1}) \\
         &\leq - \frac{\Psi(t)}{n^{3/2}}
    \end{align*}
    We now apply the multiplicative drift theorem(\cref{thm:mult_drift_tail_lengler_18}) to bound $T_1$ with $r = 3\ln n$, $s_0 = x_2(t)/x_1(t) \leq n$, $s_{\min} = (n-1)^{-1}$ and $\delta = n^{-3/2}$
    \begin{align*}
        \Probs{T_1 > 6n^{3/2}\ln n}
        =\Probs{T_1 > \frac{6\ln n}{n^{-3/2}}}
        &\leq \Probs{T_1 > \left\lceil \frac{3\ln n+\ln(n/(n-1)^{-1})}{n^{-3/2}} \right\rceil} \\
        &\leq\Probs{T_1 > \left\lceil \frac{r+\ln(s_0/s_{\min})}{\delta} \right\rceil} \leq e^{-r} = n^{-3}
    \end{align*}
\end{proof}

\subsection{Cases in which Opinion 2 wins}
\label{sec:analysis:case2}
In this section, we prove \sref{thm:main-theorem:2} of \cref{thm:main-theorem}, namely that Opinion 2 wins if $p$ is sufficiently smaller than $1-x_1(0)/x_2(0)$.
The general approach is identical to \cref{sec:analysis:case1}; given the asymmetric nature of the problem, some calculations differ slightly.
For any $t\geq0$, we call $\Delta_{\Bar{w}}=(1-p)x_2(t)-x_1(t)$ the \emph{negative weighted bias}.
Note that for $p = 1- x_1(0)/x_2(0) - \gamma$, the initial negative weighted bias is $\gamma \cdot x_2(0)$.
Analogous to \cref{lem:shifted_bias_bound:case1}, we show that the negative weighted bias does not decrease significantly for polynomial many interactions in this setting.

\begin{lemma}
\label{lem:shifted_bias_bound:case2}
Let $\textbf{x}(t_0)$ be a configuration with weighted bias $\Delta_{\Bar{w}}(t_0) \geq c_s\cdot n$.
    Let $\xi(\tau)$ be the event that $\Delta_{\Bar{w}}(t) \geq \Delta_{\Bar{w}}(t_0)/2$ for all $t\in \intcc{t_0,\ldots, t_0+\tau}$.
    Then, with probability at least $1-n^{-6}$, $\xi(T)$ holds for all $ \tau \leq \Delta_{\Bar{w}}^2(t_0)/(16 \ln n)$.
\end{lemma}
\begin{proof}[Proof sketch]
    The proof follows along the lines of that of \cref{lem:shifted_bias_bound:case1} with 
    $\Delta_{\Bar{w}}(t) = -\Delta_w(t)$.
\end{proof}

\begin{lemma}
\label{lem:drift_x1_to_zero}
     Let $\textbf{x}(t_0)$ be a configuration with weighted bias $\Delta_{\Bar{w}}(t_0) \geq c_s\cdot n$ for an arbitrary constant $c_s$. 
    Let $T_1 = \inf\{t\geq 0 ~|~ x_2(t) = 0 \}$.
    Then, $ \Probs{T_1 \leq 20 \cdot c_s^{-1} \cdot n\log n} \geq 1-n^{-2}.$
\end{lemma}
\begin{proof}
    The proof follows along the lines of that of \cref{lem:drift_x2_to_zero} with the potential function $\Psi(t) = x_1(t)/x_2(t)$.
    Recall that the idea is to calculate the expected change of the potential function $\Psi(t)$ and apply a known drift theorem.
    From \cref{lem:shifted_bias_bound:case2} and the initial size of $x_2(0)$  we get that \whp $\Delta_{\Bar{w}}(t) \geq \Delta_{\Bar{w}}(t_0)/2$.
    In particular, we also get $x_2(t) \geq c_2 \cdot n$ for some constant $c_2$.
    This allows us to bound the expected change as follows.
    
    \begin{align*}
        \MoveEqLeft\E{\Psi(t+1)-\Psi(t) ~|~ \filter_t} \\
        &=  \frac{x_1 \cdot x_2}{n^2}\left( p\Psi + (1-p)\frac{x_1-1}{x_2} \right)
        +   \frac{x_2 \cdot x_1}{n^2} \left(\frac{x_1}{x_2-1} \right)
        +   \frac{u\cdot x_1}{n^2} \left(\frac{x_1+1}{x_2} \right)\\
        &\phantom{={}}+   \frac{u \cdot x_2}{n^2} \left(\frac{x_1}{x_2+1}\right)
        +   \frac{x_1^2 + x_2^2 + (x_1+x_2+u)u }{n^2}\cdot \Psi - \Psi \\
        &=  \frac{x_1 \cdot x_2}{n^2}\left( \Psi + \frac{1-p}{x_2} \right)
        +   \frac{x_2 \cdot x_1}{n^2} \left( \Psi -\frac{x_1}{x_2} + \frac{x_1}{x_2-1} \right) 
        +   \frac{u\cdot x_1}{n^2} \left( \Psi+\frac{1}{x_2} \right) \\
        &\phantom{={}} +   \frac{u \cdot x_2}{n^2} \left( \Psi - \frac{x_1}{x_2} +\frac{x_1}{x_2+1} \right)
        +   \frac{x_1^2 + x_2^2 + (x_1+x_2+u)u }{n^2}\cdot \Psi - \Psi \\
        &= \frac{\Psi}{n^2} \left( -(1-p)x_2 + \frac{x_1 \cdot x_2}{x_2-1} + u -\frac{u \cdot x_2}{x_2+1}\right) \\
        &= -\frac{\Psi}{n^2} \left( (1-p)x_2 -x_1 - \frac{x_1}{x_2-1} - \frac{u}{x_2+1}\right) \\
        &\leq -\frac{\Psi}{n^2} \left( c_s \cdot n - \frac{n}{c_2 \cdot n -1} - \frac{n}{c_2 \cdot n +1}\right) \\
        &= -\frac{\Psi}{n} \left(c_s - \frac{1}{c_2 \cdot n -1} - \frac{1}{c_2 \cdot n +1}\right) \\
        &\leq -\frac{c_s}{2n}\cdot \Psi
    \end{align*}
\newpage    
We now apply the multiplicative drift theorem (\cref{thm:mult_drift_tail_lengler_18}) with $r = 3\ln n$, $s_0 =x_1(0)/x_2(0)\leq n$, $s_{\min} = (n-1)^{-1}$, $\delta = c_s/(2n)$ and get
\begin{align*}
    \Probs{T_1 > 12c_s^{-1} n\ln n} 
    &=\Probs{T_1 > \left\lceil \frac{12n\ln n}{c_s} \right\rceil }
    \leq \Probs{ T_1 > \left\lceil \frac{3\ln n+\ln(n/(n-1)^{-1})}{c_s/(2n)} \right\rceil } \\
    &\leq \Probs{ T_1 > \left\lceil \frac{r+\ln(s_0/s_{\min})}{\delta} \right\rceil }
    \leq e^{-r}
    = n^{-3}
\end{align*}
\end{proof}

Akin to \cref{lem:repeated-doubling}, we now consider the case $\Delta_{\Bar{w}} = o(n)$.

\begin{lemma}
    \label{lem:repeated-doubling2}
    Let $\xconfig(t_0)$ be a configuration with $\Delta_{\Bar{w}}(t_0) \geq \shBias \cdot \sqrt{n\log n}$ and let $T = \inf\set{t\geq t_0 ~|~ \Delta_{\Bar{w}}(t) \geq  n/10}$.
    Then $\Probs{T \leq (\shBias^2/6)\cdot  n\log n}\geq 1-n^{-3}$. 
\end{lemma}

\begin{proof}[Proof sketch]
    The proof follows along the lines of that of \cref{lem:repeated-doubling} for $\Delta_{\Bar{w}}$ instead of $\Delta_w$ using \cref{lem:shifted_bias_bound:case2} instead of \cref{lem:shifted_bias_bound:case1}.
\end{proof}

We now prove \sref{thm:main-theorem:2} of \cref{thm:main-theorem}.

\begin{proof}[Proof of \sref{thm:main-theorem:2}]
    Consider a configuration with $x_1(0) \in \intcc{\epsilon \cdot n, x_2}, u(0) \leq n/2$ and $p_s-p = \Omega\left(n^{-1/2}\cdot \log n\right)$.
    Then equivalently $\Delta_{\Bar{w}}(0) = \Omega(x_2 \cdot n^{-1/2} \cdot \log n)$ and $x_1+x_2 \geq n/2$. Since $x_1 \leq x_2$, we have $x_2 = \Theta(n)$ and thus $\Delta_{\Bar{w}}(0) \geq \xi \cdot \sqrt{n \log n}$ for some constant $\xi$.
    
    Let $T_a = \inf\set{t\geq 0 ~|~ \Delta_{\Bar{w}}(t) \geq n/10 }$
    and $T_b = \inf\set{t\geq 0 ~|~ x_1(t) = 0 }$.
    By \cref{lem:repeated-doubling2}, we have \whp $T_a = O(n\log n)$ and then by \cref{lem:drift_x1_to_zero} $T_b = T_a + O(n \log n)$.
    
    Note that at no time all agents can be undecided, since
    the last agent with Opinion 2 cannot encounter Opinion 2.
    Therefore, at time $T_b$, at least one agent with Opinion 2 exists.

    With $x_1(T_b) = 0$, the process simplifies to a single productive rule: $\delta(2,\bot) = 1$. 
    Let $T_2 = \inf\set{t \geq T_b ~|~ x_2(t) = n}$.
    Assuming that $T_b < \infty$ and $x_2(T_2) = 1$, we have $T_2 \leq T_b + 6 n \log n$ with probability at least $1-n^{-2}$ (see e.g., \cite{berenbrink2021loosely}).
    It is easy to see that $x_2(T_b)=1$ gives an upper bound for $T_2$.
    The statement then follows from the union bound.
\end{proof}

\end{document}